\documentclass[journal,10pt]{IEEEtran}

\IEEEoverridecommandlockouts
\normalsize
\usepackage{cite}
\usepackage{url}
\usepackage{hyperref} 
\hypersetup{colorlinks=true, linkcolor=red, citecolor=red, urlcolor=blue} 
\usepackage{amssymb}
\usepackage{amsmath}
\usepackage{array}
\usepackage{amsthm}
\usepackage{mathptmx}
\usepackage{mathtools, cuted}
\usepackage{algpseudocode}
\usepackage[ruled, vlined, linesnumbered]{algorithm2e}
\usepackage{graphicx}
\usepackage{xcolor}
\usepackage{bbm}
\usepackage{multicol}
\usepackage{setspace}

\allowdisplaybreaks

\newtheorem{thm}{Theorem}

\newtheorem{prop}{Proposition}
\newtheorem{defn}{Definition}
\newtheorem{rem}{Remark}

\newcommand{\eqz}{\!\!}
\newcolumntype{C}[1]{>{\centering\arraybackslash}p{#1}}

\begin{document}
\title{Efficient Content Delivery in User-Centric and Cache-Enabled Vehicular Edge Networks with Deadline-Constrained Heterogeneous Demands}

\author{\vspace{-0.05in}
Md Ferdous Pervej, \IEEEmembership{Graduate Student Member, IEEE}, Richeng Jin, \IEEEmembership{Member, IEEE}, Shih-Chun Lin, \IEEEmembership{Member, IEEE}, and Huaiyu Dai, \IEEEmembership{Fellow, IEEE}
\thanks{M.F. Pervej, S.-C. Lin and H. Dai are with the Department of Electrical and Computer Engineering, NC State University, Raleigh, NC 27695, USA. Email: \{mpervej, slin23, hdai\}@ncsu.edu.} 
\thanks{R. Jin was with the Department of Electrical and Computer Engineering, NC State University, and is now with the College of Information Science and Electronic Engineering, Zhejiang University, Hangzhou 310058, China. Email: \{rjin2@ncsu.edu; richengjin@zju.edu.cn\}.}
\vspace{-0.35in}
}

\markboth{DRAFT}
{}
\maketitle

\begin{abstract}
Modern connected vehicles (CVs) frequently require diverse types of content for mission-critical decision-making and onboard users' entertainment.
These contents are required to be fully delivered to the requester CVs within stringent deadlines that the existing radio access technology (RAT) solutions may fail to ensure.
Motivated by the above consideration, this paper exploits content caching in vehicular edge networks (VENs) with a software-defined user-centric virtual cell (VC) based RAT solution for delivering the requested contents from a proximity edge server.
Moreover, to capture the heterogeneous demands of the CVs, we introduce a preference-popularity tradeoff in their content request model.
To that end, we formulate a joint optimization problem for content placement, CV scheduling, VC configuration, VC-CV association and radio resource allocation to minimize long-term content delivery delay.
However, the joint problem is highly complex and cannot be solved efficiently in polynomial time. 
As such, we decompose the original problem into a cache placement problem and a content delivery delay minimization problem given the cache placement policy.
We use deep reinforcement learning (DRL) as a learning solution for the first sub-problem. 
Furthermore, we transform the delay minimization problem into a priority-based weighted sum rate (WSR) maximization problem, which is solved leveraging maximum bipartite matching (MWBM) and a simple linear search algorithm.
Our extensive simulation results demonstrate the effectiveness of the proposed method compared to existing baselines in terms of cache hit ratio (CHR), deadline violation and content delivery delay.
\end{abstract}

\vspace{-0.05in}

\IEEEpeerreviewmaketitle

\vspace{-0.1in}
\section{Introduction}
\IEEEPARstart{A}{dvanced} driver-assistance systems (ADAS) and infotainment are two premier features for modern connected vehicles (CVs). 
With advanced radio access technologies (RATs), delivering the Society of Automotive Engineers (SAE) level $5$ automation on the road seems more pragmatic day by day. 
Different government organizations - such as the U.S. Department of Transportation's National Highway Traffic Safety Administration in the United States \cite{NHTSA_1}, the Department for Transport in the U.K. \cite{DfT_UK}, etc., set firm regulations for the CVs to ensure public safety on the road.
For swift decision-making to satisfy the safety requirements, the CVs need fast, efficient, and reliable communication and data processing.
As such, an efficient vehicular edge network (VEN) must ensure uninterrupted and ubiquitous wireless connectivity on the road. 
Note that a VEN is an edge network that mainly focuses on communication among vehicles and/or between vehicles and infrastructure \cite{li2021federated}.
To deliver above services, the VEN demands advanced machine learning (ML) tools for resource management complementary to a RAT solution, such as the 5G new-radio (NR) vehicle-to-everything (V2X) communication \cite{liu20206g}.

With increased automation, in-car entertainment is also becoming a priority for modern CVs \cite{8734737}.
Modern CVs are expected to have many new features, such as vehicular sensing, onboard computation, virtual personal assistant, virtual reality, vehicular augmented reality, autopilot, high-definition (HD) map collection, HD content delivery, etc., \cite{8515151, chandupatla2020augmented} that are interconnected for both ADAS and infotainment.
For these demands, by exploiting the emerging content caching \cite{Pervej_UPLEdge}, the centralized core network can remarkably gain by not only ensuring local content distribution but also lessening the core network congestion \cite{8367785, Pervej_AIACEC}.
As such, VENs can reduce end-to-end latency significantly by storing the to-be-requested contents at the network edge \cite{9552606}, which is vital for the CVs' mission-critical delay-sensitive applications.
A practical RAT on top of content caching can, therefore, bring a promising solution for SAE level $5$ automation on the road.
Moreover, owing to these multifarious requirements, it is also critical to explore the efficacy of content caching with limited cache storage and different types of content classes, each class with multiple contents, in the content library.

For diverse applications, such as mobile broad bandwidth and low latency (MBBLL), massive broad bandwidth machine-type (mBBMT), massive low-latency machine-type (mLLMT) communications, etc., the CVs urgently need an efficient RAT solution \cite{9382020}.
In the meantime, regardless of the applications, the VEN must ensure omnipresent connectivity to the CVs and deliver their requested contents timely.
The so-called user-centric networking \cite{9019853, Pervej_throughput, zhou2021user, 9462895} is surging nowadays with its ability to shift network resources towards network edge.
Note that a user-centric approach is based on the idea of serving users by creating virtual cells (VCs) \cite{9134799,8334916,8088603}.
While the network-centric approach serves a user from only one base station, the user-centric approach enables serving a user from a VC that may contain multiple transmission points \cite{9134799,8334916,8088603}.
The latter approach can, thus, not only provide ubiquitous connectivity but also provide higher throughput with minimized end-to-end latency for the end-users \cite{Pervej_EE}.
As such, a user-centric approach can combat the frequent changes in received signal strength - often experienced in VENs due to high mobility, by ensuring multipoint data transmission and receptions.

While the user-centric networking approach can bring universal connectivity and MBBLL/mBBMT/mLLMT solutions for the CVs, it induces a more complex network infrastructure. 
To ensure multipoint data transmission and reception, efficient baseband processing is required.
Moreover, as the traditional hardware-based and closed network-centric approach is inflexible, the user-centric approach demands the use of software-defined networking \cite{6994333}, which can offer more efficient and agile node associations and resource allocations in the user-centric approach.
With proper system design, it is possible to create VCs with multiple low-powered access points (APs) to ensure that the throughput and latency requirements of the CVs are satisfied.
Moreover, amalgamating content caching with the user-centric RAT solution can indeed ensure timely payload delivery for stringent delay-sensitive application requirements of modern CVs.
However, this requires a joint study for - content placement, CV scheduling, VC formulation, VC association with the scheduled CV, and radio resource allocation of the APs in the VCs.

\vspace{-0.15in}
\subsection{Related Work}
\noindent
In literature, there exist several works \cite{huang2021delay,nan2021delay, 8993754, 8998330, 8998397, zhang2020smart, 9129007, 9417383}
that considered cache-enabled VENs from the traditional network-centric approach. 
Huang \textit{et al.} proposed a content caching scheme for the Internet of vehicles (IoVs) in \cite{huang2021delay}. 
They developed a delay-aware content delivery scheme exploiting both vehicle-to-infrastructure (V2I) and vehicle-to-vehicle (V2V) links. 
The authors minimized content delivery delays for the requester vehicles by jointly optimizing cache placement and vehicle associations.
Nan \textit{et al.} also proposed a delay-aware caching technique assuming that the vehicles could either a) decide to wait for better delivery opportunities, or b) get associated with the roadside unit (RSU) that has the content, or c) use one RSU as a relay to extract the content from the cloud in \cite{nan2021delay}.
The authors exploited deep reinforcement learning (DRL) to minimize content delivery cost.
However, these assumptions are not suitable for CVs because time-sensitivity plays a crucial role in the quick operation of CVs.
\cite{8993754} proposed quality-of-service ensured caching solution by bounding the content into smaller chunks.

Dai \textit{et al.} leveraged blockchain and DRL to maximize caching gain \cite{8998330}.
Lu \textit{et al.} proposed a federated learning approach for secure data sharing among the IoVs \cite{8998397}.
However, \cite{8998330, 8998397} assumed that the data rate is perfectly known without any proper resource allocations for the RAT.
Zhang \textit{et al.} addressed proactive caching by predicting user mobility and demands in \cite{zhang2020smart}.
Similar prediction-based modeling has also been extensively studied in \cite{Pervej_UPLEdge,8531745,malik2020personalized}.
Moreover, \cite{zhang2020smart} only analyzed cache hit ratio without incorporating any underlying RAT.
Fang \textit{et al.} considered a static popularity-based cooperative caching solution for roaming vehicles, which assumed constant velocity and downlink data rate and minimized content extraction delay \cite{9129007}.
Liu \textit{et al}. considered coded caching for a typical heterogeneous network with one macro base station (MBS) overlaid on top of several RSUs \cite{9417383}. 
Vehicles trajectory, average residence time within RSU's coverage, and system information were assumed to be perfectly known to the MBS in \cite{9417383}.
Owing to the time-varying channel conditions in VENs, the authors further considered a two-time scaled model.
Particularly, they assumed that content requests only arrive at the large time scale (LTS) slot, whereas MBS could decide to orchestrate resources in each small time scale (STS) slot - within the LTS slot.
However, although \cite{9417383} assumed LTS and STS considering time-varying wireless channels, it did not consider any communication model.
Therefore, the study presented in \cite{9417383} did not reflect delay analysis in VENs.


The study presented in \cite{huang2021delay,nan2021delay, 8993754, 8998330, 8998397, zhang2020smart, 9129007, 9417383} mostly considered that the content catalog consist of a fixed number of contents from a single category.
In reality, each content belongs to a certain category, and the catalog consists of contents from different categories.
Besides, these studies mainly assumed that the users request contents based on popularity.
However, each CV may have a specific need for a particular type of content.
For example, some CVs may need to have frequent operational information, whereas other CVs may purely consume entertainment-related content.
Therefore, a VEN shall consider individual CV's preference, as well as the global popularity.

Some literature also exploited user-centric RAT solutions for VENs \cite{Pervej_throughput,Pervej_EE,8334916,9134799,8088603,9275345}. 
Considering the high mobility of the vehicles, \cite{Pervej_throughput} proposed an approach for user-centric VC creation and optimized resource allocation to ensure maximized network throughput. 
A power-efficient solution for the VC of the VENs was also proposed in \cite{Pervej_EE}.
Lin \textit{et al.} proposed heterogeneous user-centric (HUC) clustering for VENs in \cite{9275345}.
Particularly, the authors considered creating HUC using both traditional APs and vehicular APs.
The goal of \cite{9275345} was to study how HUC migration helps in VEN.
Considering both horizontal handover (HO) and vertical HO, \cite{9275345} studied the tradeoff between throughput and HO overhead. 
Xiao \textit{et al}. showed that dynamic user-centric virtual cells could be used to multicast the same message to a group of vehicles in \cite{9134799}. 
Particularly, \cite{9134799} assumed that a group of vehicles could be considered as a hotspot (HS).
If all vehicles inside the HS are interested in the same multicasted message, multiple APs could formulate a VC to serve the HS.
\cite{9134799} optimized power allocation to balance the signal-to-interference-plus-noise ratio for the vehicles in the HS.
Shahin \textit{et al.} also performed similar studies in \cite{8334916,8088603}.
Instead of serving a single user, they created HS for V2X broadcast groups. They then maximized the total active HS in the network using admission control, transmission weight selection and power control \cite{8334916,8088603}.

\vspace{-0.15in}
\subsection{Motivations and Our Contributions}
\noindent 
As ubiquitous connectivity is essential for CVs, the existing RAT solutions may not be sufficient to meet the strict requirements of CVs for higher automation.
Existing literature shows that VC-based user-centric networking can bring additional burdens that need rigorous studies, such as mobility and HO management \cite{9275345}.
Moreover, as multicasting delivers a common signal, the study presented in \cite{9134799,8334916,8088603} is not suitable for CV-specific independent data requirements in delay-sensitive applications.
However, an alternative software-defined networking approach with advanced ML algorithms can potentially bring the RAT solution \cite{Pervej_throughput, Pervej_EE,lin2021sd}.
Moreover, \cite{Pervej_throughput, Pervej_EE} considered that all APs could serve all users, which may not be possible due to limited coverage and other resource constraints.
Inspired by the user-centric VC-based studies \cite{Pervej_throughput,Pervej_EE,8334916,9134799,8088603,9275345}, our proposed VEN can deploy a close proximity edge server that acts as the software-defined controller.
The to-be-requested contents can be prefetched through the edge servers to ensure local delivery.
Besides, multiple low-powered APs can be placed as RSUs.
The controller can determine the user-centric VC configuration and the corresponding resource orchestration to meet the requirements of the CVs by controlling these APs.

\begin{table}[!t]
\caption{Important Notations Utilized in This Paper}
	\fontsize{7}{7}\selectfont
	\centering
	\begin{tabular}{C{1.8cm}|| C{6.2 cm} }
		\hline 
		\textbf{Parameter} & \textbf{Definition}  \\ \hline
		$\mathcal{B}, B, b$ & Set of APs, total number of APs, $b^{\text{th}}$ AP \\ \hline
		$\mathcal{U}, U, u$ & Set of CVs, total number of CVs, $u^{\text{th}}$ CV \\ \hline
		$W_{\text{max}}$, $W(t)$ & Maximum possible VCs with $B$ APs, total created VCs in slot $t$ \\ \hline
		$\mathcal{B}_{vc}(W(t))$,\! $A_{W(t)}$, $\mathcal{B}_{vc}^{a}(W(t))$ & All VC configurations set, total possible VC configurations in $\mathcal{B}_{vc}(W(t))$, VC sets under $a^{\text{th}}$ configuration \\ \hline
		$VC_a^i$ & $i^{\text{th}}$ VC of $\mathcal{B}_{vc}^a(W(t))$ \\ \hline
		$\mathrm{I}_{b}^{i,a}(t) $ & Indicator function that defines whether AP $b$ is in $VC_a^i$ \\ \hline
		$\mathrm{I}_u(t)$ & Indicator function that defines whether CV $u$ is scheduled at slot $t$ \\ \hline 
		$\mathrm{I}_{u}^{i,a}(t)$ & Indicator function that defines whether VC $VC_a^i$ is selected for user $u$ at time $t$ \\ \hline
		$\bar{Z}$, $Z$, $z$, $\omega$ & Total network bandwidth, total orthogonal pRB, $z^{\text{th}}$ pRB, size of the pRBs \\ \hline 
		$\mathrm{I}^{b,u}_{z}(t)$ & Indicator function that defines whether pRB $z$ is assigned to AP $b$ to serve CV $u$ at time $t$ \\ \hline
		$\psi_b^u(t), \tau_b^u(t)$, $\breve{\mathbf{h}}_{b}^{u,z}(t)$, $\mathbf{h}_{b}^{u,z}(t)$ & Large scale fading, log-Normal shadowing, fast fading channel response, entire channel response, respectively, from AP $b$ to CV $u$ during slot $t$ over the $z^{\text{th}}$ pRB \\ \hline
		$\mathbf{H}_u^b(t)$  & Stacked $\mathbf{h}_b^{u,z}$s over all pRBs for CV $u$ and AP $b$ during slot $t$ \\ \hline 
		$x_b^u(t)$, $\mathbf{s}_b^u(t)$, $\mathbf{w}_b^{u,z}$ & Intended signal, symbol, and beamforming vector of AP $b$ for CV $u$, respectively \\ \hline
		$P_b$ & Transmission power of AP $b$ \\ \hline
		$y_u^z(t)$, $\Gamma_u^z(t)$ & Downlink received signal and downlink SNR at CV $u$ over pRB $z$, respectively, during slot $t$ \\ \hline 
		$\kappa$ & Transmission time interval \\ \hline
		$R_u(t)$ & Downlink achievable rate at CV $u$ during slot $t$ \\ \hline  
		$\mathcal{C}$, $C$, $c$ & Content class set, total content class, $c^{\text{th}}$ content class \\ \hline 
		$\mathcal{F}_c$, $F$, $f_c$, $\mathcal{F}$ & Content set in class $c$, total content in a class, $f^{\text{th}}$ content of class $c$, entire content library \\ \hline 
		$\mathcal{G}_{f_c}$, $G_c$, $g_{f_c}$ & Set of the content features, total number of features, $g_{f_c}^{\text{th}}$ feature, respectively, of content $f$ of class $c$ \\ \hline
		$\Upsilon$, $n$ & DoI, cache (re)-placement or DoI change counter \\ \hline 
		$\Theta_u^t$ & Bernoulli random variable that defines whether CV $u$ places a content request at time $t$ \\ \hline
		$\Psi_t$ & Total content requests from all CVs during time $t$ \\ \hline
		$\mathrm{I}_u^{f_c} (t)$  & Indicator function that defines whether CV $u$ requests content $f_c$ during time $t$ \\ \hline
		$\Omega_{f_c}^{f_c^{'}}$ & Cosine similarity index of content $f_c$ and $f_{c}^{'}$ \\ \hline 
		$S$, $\Lambda$, $\Lambda^c$ & content size, cache storage size of edge server, cache storage to be filled with content from class $c$ \\ \hline 
		$\mathrm{I}_{f_c}(n)$ & Indicator function that defines whether content $f_c$ is stored during cache placement counter $n$ \\ \hline
		$\epsilon_u$ & Highest probability for content exploitation of CV $u$\\ \hline
		$p_c^u$ & CV $u$'s probability of selecting content class $c$ \\ \hline
		$p_c^{f_c}$ & Global popularity of content $f_c$ of class $c$ \\ \hline
		$d_{u,f_c}^{m,t}$, $d_{u,f_c}^{q,t}$, $d_{u,f_c}^{s,t}$  & Content extraction delay from cloud, wait time of $\mathrm{I}_u^{f_c}(t)$ before being scheduled, transmission delay \\ \hline 
		$d_f^{max}$, $\hat{\mathrm{d}}_{u}^{f_c}(t)$ & Maximum allowable delay by the CV, hard-deadline for the edge server to completely offload the requested content \\ \hline 
		$\bar{d}(t)$ & Average delays for all $\mathrm{I}_u^{f_c}(t)$s \\ \hline
		$\eqz \eqz \eqz \eqz m_{ca}, m_{joint}(n)$  & Cache placement action, possible action space\\ \hline
		$\mathbbm{1}_{\mathrm{I}_u^{f_c}}(t)$ & Cache hit event for $\mathrm{I}_u^{f_c}(t)$ \\ \hline
		$h(t)$, CHR$(t)$ & Total cache hit during slot $t$, cache hit ratio during slot $t$\\ \hline
		$\pi_{ca}$ & Cache placement policy of the edge server \\ \hline 
		$\mathcal{T}_{\text{SoI}}^t$ & Slots of interests during slot $t$\\ \hline
		$\mathrm{T}_{u,\text{rem}}^{t-d_f^{max}+\zeta}$, $\mathrm{P}_{u,\text{rem}}^{t-d_f^{max}+\zeta}$ & Remaining time to the deadline and payload for the requested content in slot $t-d_f^{max}+\zeta$ at the current slot $t$ \\ \hline
		$\mathcal{U}_{\text{val}}^t$, $\mathcal{T}_{\text{rem}}^t$, $\mathcal{P}_{\text{rem}}^t$ & Valid CV set during slot $t$, and their minimum remaining deadline and payload sets, respectively \\ \hline  
		$\phi_u(t)$ & Normalized weights of the CVs in valid CV set during slot $t$ \\ \hline 
		$\mathcal{U}_{\text{sch}}^t$ & Scheduled CV set during slot $t$ \\ \hline
		$\breve{R}(t)$ & Weighted sum rate of the VEN during slot $t$ \\ \hline 
		$\mathbf{F}^{top}(n)$ & Top-most popular and their $\Lambda^c$-top similar contents matrix \\ \hline
		$\mathbf{P}_{req}^u (n)$ & Content-specific requests history matrix of $u$ in past DoI \\ \hline
		$\mathbf{P}_{hit}^u (n)$ & Content-specific local cache hit history matrix of $u$ in past DoI \\ \hline
		$\mathbf{P}_{f}(n)$ & Measured popularity of contents during $n$ based on past DoI \\ \hline
		$x_{ca}^{n}$, $r_{ca}^{n}$ & State and instantaneous reward of the edge server during $\pi_{ca}$ learning \\ \hline
		$\pmb{\theta}_{ca}$, $\pmb{\theta}_{ca}^{-}$ & Online DNN and offline DNN of the edge server for learning the CPP \\ \hline
		\eqz $mem_{ca}$, $mem_{ca}^{max}$ & Edge server's memory buffer, maximum length of $mem_{ca}$ for learning $\pi_{ca}$ \\ \hline
		$G$, $\mathbf{R}_t$ & Bipartite graph, weighted data rate matrix \\ \hline
		$e(b,z)$, $R_t[b,z]$ & Edge connecting vertex $b$ and $z$, and corresponding weights \\ \hline
		$R_u^{bit} (t) $ & Possible transmissible bits for CV $u$ during slot $t$ \\ \hline
	\end{tabular}
	\label{System_Parameters}
\end{table}

In comparison to the above studies, in this work, we have considered a practical communication model, introduced a preference-popularity tradeoff in content request models, considered a multi-class content catalog, introduced a new VC formation strategy that exploits all possible ways of partitioning the low-powered APs, introduced a duration of interest (DoI) for which the edge server cannot update the cache storage due to practical hardware and overhead constraints, and devised a joint cache placement and user-centric RAT solution.
Particularly, our contributions are
\begin{itemize}
	\item Considering the stringent requirements of the regulatory organizations, we propose a new software-defined user-centric RAT solution that partitions the low-powered APs to form VC, and provides ubiquitous and reliable connectivity to the CVs on the road.
	\item To ensure fast decision-making for mission-critical operations and uninterrupted onboard entertainment, we exploit content prefetching at the edge server, with multiple classes in the content catalog, while introducing preference-popularity tradeoff into individual content requests owing to the CVs' heterogeneous preferences. 
    Moreover, we introduce a DoI for which the cached contents remain idle due to practical limitations and leverage our proposed RAT solution to deliver the requested contents within a hard deadline.
	\item We introduce a joint content placement, CV scheduling, VC configuration, CV-VC association and radio resource allocation problem to minimize content delivery delays.  
	\item To tackle the grand challenges of the joint optimization problem, we decompose it into a cache placement sub-problem and a delay minimization sub-problem - given the cache placement policy. We propose a novel DRL solution for the first sub-problem. We then transform the second sub-problem to a weighted sum rate (WSR) maximization problem due to practical limitations and solve the transformed problem using maximum weighted bipartite matching (MWBM) and a simple linear search algorithm.
	\item Through analysis and simulation results, we verify that our proposed solution achieves better performance than the existing baselines in terms of cache hit ratio (CHR), deadline violation and content delivery delay.
\end{itemize}

The rest of the paper is organized as follows: Section \ref{Sys_Model} introduces our proposed software-defined user-centric system model.
Section \ref{cacheModelingSection} presents the caching model, followed by the joint problem formulation in Section \ref{I2V_Com_Delay_Model}. 
Problem transformations are detailed in Section \ref{probTransformationSec}, followed by our proposed solution in Section \ref{problemSolutionSec}. 
Section \ref{resultDiscussSection} presents extensive simulation results and discussions. Finally, Section \ref{conclusion} concludes the paper. 
The important notations are listed in Table \ref{System_Parameters}.

\vspace{-0.1in}
\section{Sofware-Defined User-Centric Communication Model}
\label{Sys_Model}

\subsection{Communication System Model}
\label{comSysModel}
\noindent 
This paper considers a software-defined cache-enabled VEN.
An edge server - controlled by a software-defined controller, is placed in proximity to the edge CVs. 
The edge server has dedicated radio resources with limited local cache storage and is connected to the cloud.
Several low-powered APs are deployed as RSUs to provide omnipresent wireless connectivity to the CVs. 
These APs are connected to the edge server with high-speed wired links.
The software-defined centralized controller can control the edge server and perform user scheduling, node associations, precoding, channel estimations, resource allocations, etc. 
In other words, the edge server acts as the baseband unit.
Besides, unlike the legacy system models, we consider a user-centric approach that uses multiple APs to serve the scheduled CVs. 
These APs are used as RSUs that only perform radio transmissions over the traditional Uu interface \cite{3GPP_TR_38_886}.
Denote the vehicle and AP set by $\mathcal{U}=\{u\}_{u=1}^U$ and $\mathcal{B} = \{b\}_{b=1}^B$, respectively. 
The VEN operates in slotted times. 
Given $B$ APs at fixed locations, unlike the traditional network-centric approach, the proposed VEN partitions the AP set $\mathcal{B}$ into $W(t) \leq B$ subsets of APs at each slot $t$.
Without loss of generality, we define each subset as a VC.
The proposed VEN can assign such a VC to a scheduled CV.

Let there be a set $\mathcal{A}(W(t))=\{a\}_{a=1}^{A_{W(t)}}$ that defines the possible ways to partition the $B$ APs into $W(t)$ subsets of APs, i.e., VCs.
Denote the $i^{\text{th}}$ VC of $a \in \mathcal{A}\left(W(t)\right)$ by $VC_a^i \subset \mathcal{B}$ and the set of VCs by $\mathcal{B}_{vc}^a(W(t))=\{VC_a^i\}_{i=1}^{W(t)}$.
Furthermore, for each $a \in \mathcal{A}(W(t))$, the VC must obey the following rules:
\begin{subequations}
\small
	\begin{align}
		VC_a^i &\neq \emptyset, ~ \forall a \text{ and } i, \label{VC_Cons1}\\
		VC_a^i \bigcap VC_a^{i'} &= \emptyset, ~ \forall a \text{ and } i \neq i', \label{VC_Cons2}\\ 
		\bigcup\nolimits_{i=1}^W VC_a^i &= \mathcal{B}, ~\forall a. \label{VC_Cons3}
	\end{align}
\end{subequations}
The first rule in (\ref{VC_Cons1}) means that the $VC_a^i$s must not be empty, while the second rule in (\ref{VC_Cons2}) means the subsets are mutually exclusive.
Besides, the last rule in (\ref{VC_Cons3}) represents that the union of the $VC_a^i$s must yield the original AP set $\mathcal{B}$.
When $W(t)=3$ and $B=6$, for one possible $a \in \mathcal{A}(W(t))$, the VC set $\mathcal{B}_{vc}^a (W(t))$ is shown in Fig. \ref{system_model} by the filled ovals.

Then, for a given $W(t)$, the VEN can formulate the total VC configuration pool $\mathcal{B}_{vc}(W(t)) = \left\{\mathcal{B}_{vc}^a(W(t))\right\}_{a=1}^{A_{W(t)}}$.
Moreover, the VEN can partition the $B$ APs into $W(t)$ VCs following the above rules in $A_{W(t)}=W(t)! B_{vc} (B,W(t))$ ways, where $B_{vc}(B,W(t))$ is calculated as 
\begin{equation}
\label{VC_Stirling_No}
\small
	B_{vc}(B,W(t)) = \frac{1}{W(t)!} \sum\nolimits_{\bar{w}=1}^{W(t)} \left(-1\right)^{\bar{w}} { \binom{W(t)} {\bar{w}} } \left( W(t) - \bar{w} \right)^B\eqz.\eqz
\end{equation}
Note that $B_{vc}(B,W(t))$ is commonly known as the Stirling number of the second kind \cite{graham1989concrete}.

To this end, as the VCs contain different AP configurations, denote the VC and AP mapping by
\begin{align}
\small
	\mathrm{I}_{b}^{i,a} (t) &= \begin{cases}
		$1$, & \text{if AP $b$ is in $VC_a^i$ during slot $t$},\\
		$0$, & \text{otherwise}, 
	\end{cases}. \label{AP_VC_Mapping} 
\end{align} 
The edge server selects the total $W(t)$ number of VCs to form and their configuration $\mathcal{B}_{vc}^a(W(t)) \in \mathcal{B}_{vc}(W(t))$. 
The VCs $VC_a^i \in \mathcal{B}_{vc}^a(W(t))$ are assigned to serve the scheduled CVs.
Denote the CV scheduling and VC-CV association decisions by the following indicator functions:
\begin{align}
\small
	\mathrm{I}_{u} (t) &= \begin{cases}
		$1$, & \text{if CV $u$ is scheduled in slot $t$},\\
		$0$, & \text{otherwise}, 
	\end{cases}. \label{AN_Schedule}  \\
	\mathrm{I}_{u}^{i,a}(t) &= \begin{cases}
		1, & \text{if $VC_a^i$ is selected for CV $u$ in slot $t$},\\
		0, & \text{otherwise},
	\end{cases}. \label{CV_VC_Association}
\end{align}
Note that (\ref{CV_VC_Association}) means that if $\mathrm{I}_{u}^{i,a}(t)=1$, i.e., $VC_a^i$ is assigned to CV $u$, then the CV is connected to all APs inside this $VC_a^i$.

We contemplate that the VEN operates in frequency division duplex (FDD) mode and has a fixed $\bar{\mathrm{Z}}$ Hz bandwidth.
The edge server uses this dedicated $\bar{\mathrm{Z}}$ Hz bandwidth and divides it into $Z$ orthogonal physical resource blocks (pRBs). 
Let the set of the orthogonal pRBs be $\mathcal{Z}=\{z\}_{z=1}^Z$ for the downlink infrastructure-to-vehicle (I2V) communication. 
Denote the size of a pRB by $\omega$, while we introduce the following indicator function for the pRB allocation
\begin{align}
\small
	\mathrm{I}_{z}^{b,u} (t) &= \begin{cases}
		$1$, & \text{if pRB $z$ is assigned to AP $b$ when} \\
		& \qquad \mathrm{I}_{u}^{i,a}(t)=1  \text{ and } \mathrm{I}_{b}^{i,a}(t)=1,\\
		$0$, & \text{otherwise}, 
	\end{cases}. \label{RB_Allocation_V2I}
\end{align}

\begin{figure}[!t]
	\centering
	\includegraphics[trim=20 10 20 10, clip, width=0.5\textwidth, height=0.28\textheight]{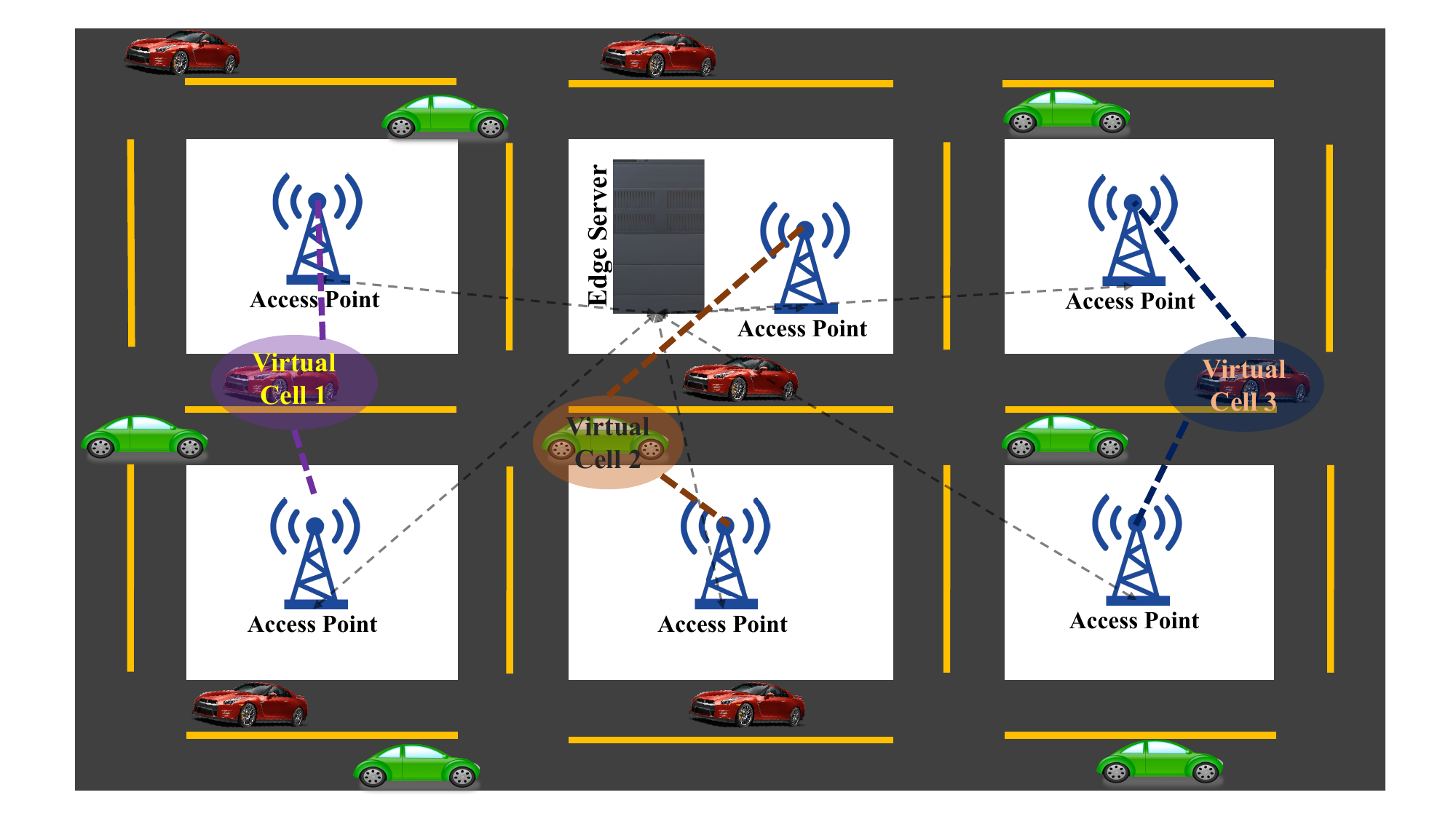}
	\caption{Proposed user-centric cache-enabled vehicular edge network}
	\label{system_model}
\end{figure}

\vspace{-0.15in}
\subsection{Communication Channel Modeling}
\noindent
We consider single antenna CVs whereas, the APs are equipped with $L>1$ antennas. 
Let us denote the channel response at a CV $u$ from the AP $b$, over pRB $z$, as follows:
\begin{equation}
	\begin{aligned}
		\mathbf{h}_{b}^{u, z}(t) &= \sqrt{\psi_b^u(t)} \tau_b^u(t) \breve{\mathbf{h}}_{b}^{u, z} (t) \in \mathbb{C}^{L \times 1},
	\end{aligned}
\end{equation}
where $\sqrt{\psi_b^u(t)}$, $\tau_b^u(t)$ and $\breve{\mathbf{h}}_b^{u,z}(t) = [h_{b,1}^{u,z}(t), \dots, h_{b,L}^{u,z}(t)]^T \in \mathbb{C}^{L \times 1}$ are large scale fading, log-Normal shadowing and fast fading channel responses from the $L$ antennas, respectively.
Besides, $h_{b,l}^{u,z}(t)$ is the $l^{\text{th}}$ row of $\breve{\mathbf{h}}_b^{u,z}(t)$ that denotes the channel between $u$ and the $l^{\text{th}}$ antenna of AP $b$ at time $t$ over pRB $z$.
Note that we consider the urban macro (UMa) model \cite{3GPP_TR_38_901} for modeling the path losses following $3$rd Generation Partnership Project ($3$GPP) standard \cite{3GPP_TR_38_886}.
Then, for all pRBs in the system, we express the wireless channels from AP $b$ to CV $u$, at time $t$, as $ \mathbf{H}_u^b (t) = \left[\mathbf{h}_{b}^{u, 1}(t), \dots, \mathbf{h}_{b}^{u, Z}(t) \right] \in \mathbb{C}^{L \times Z} $.
We consider the edge server has perfect channel state information (CSI)\footnote{Although channel reciprocity does not hold in FDD, the edge server can use some feedback channels to estimate the CSI.} and all transceivers can mitigate the Doppler effect.

\vspace{-0.1in}
\subsection{Transceiver Modeling}
\noindent
The transmitted signal at the AP $b$ for CV $u$ is 
	$\mathbf{s}_b^u(t) = \sqrt{P_{b}} x_{b}^{u} (t) \mathbf{w}_b^{u,z}(t) \in \mathbb{C}^{L \times 1}$,
where $P_b$ is the transmission power of AP $b$.
Besides, $x_b^u(t)$ and $\mathbf{w}_{b}^{u,z} (t) \in \mathbb{C}^{L \times 1}$ are the unit powered intended signal and corresponding precoding vector over pRB $z$, respectively, of the AP $b$ for $u$ during slot $t$. 
Then, the transmitted signals for CV $u$ from all APs is	$\mathbf{s}_u(t) = \left[\mathbf{s}_1^u(t), \dots, \mathbf{s}_B^u(t) \right]^T \in \mathbb{C}^{L \times B}$.
Moreover, as each AP transmits over orthogonal pRBs, the proposed VEN does not have any interference. 
To this end, we calculate the received signal at CV $u$, over pRB $z$, as
\begin{equation}
\label{V2I_Downlink_Rx_Sig}
\begin{aligned}
		&y_{u}^z(t) \eqz = \eqz \eqz \sum_{i=1}^{W(t)} \eqz \mathrm{I}_{u}^{i,a}(t) \! \Big(\! \sum_{b=1}^B \eqz \mathrm{I}_{b}^{i,a}(t) \cdot \mathrm{I}_{z}^{b,u}(t) \eqz \left[{\mathbf{h}_b^{u,z}(t)}^H \mathbf{s}_b^u(t) + \eta_b^u(t) \right] \eqz\Big)\!,\eqz\eqz\eqz \eqz\eqz\eqz
	\end{aligned}
\end{equation}
where $\eta_b^u(t) \sim CN\left(0, \sigma_b^2 \right)$ is zero mean circularly symmetric Gaussian distributed noise with variance $\sigma^2$.
The corresponding signal-to-noise ratio (SNR), over pRB $z$, is 
\begin{equation}
\label{V2I_SINR}
	\Gamma_{u}^z \left(t\right)  =  \frac{\sum_{i=1}^{W(t)} \mathrm{I}_{u}^{i,a}(t)  \big(\!\sum_{b=1}^B \! \mathrm{I}_{b}^{i,a}(t) \cdot \mathrm{I}_{z}^{b,u}(t) \! \big[ P_b \big\vert {\mathbf{h}_{b}^{u,z} (t)}^H \mathbf{w}_b^{u,z}(t) \big\vert^2 \big] \big)} { \omega \big[\sum_{i=1}^{W(t)} \mathrm{I}_{u}^{i,a}(t) \cdot \big(\sum_{b=1}^B \mathrm{I}_{b}^{i,a}(t) \cdot \sigma_b^2 \big) \big] }. 
\end{equation}
Therefore, the total downlink achievable capacity for CV $u$ is \vspace{-0.1 in}
\begin{equation}
	\label{V2I_Data_Rate}
	R_{u}\left(t\right) = \sum\nolimits_{z=1}^Z \omega \cdot \log_2 \left(1 + \Gamma_{u}^z \left(t\right) \right).
\end{equation}

\vspace{-0.05in}
\section{Edge Caching Modeling}
\label{cacheModelingSection}

\subsection{Definitions and Assumptions}
\noindent
To avoid cross-domain nomenclature, we present necessary terms and our assumptions in the following.
\vspace{-0.05in}
\begin{defn}[Content]
	The source file that the CVs request is defined as content.
	These files can contain CVs' operational information, geographic information, map/navigation information, weather conditions, compressed file with sensory data, local news, video/audio clips, etc.
\end{defn}
\vspace{-0.1in}
\begin{defn}[Content Class]
	Each content belongs to a class that defines the type/category of the content.
	Let there be $F$ contents in each class and the set of the content class be $\mathcal{C}=\{c\}_{c=1}^C$, where $C \in \mathbb{Z}^{+}$. 
	Denote the content set of class $c$ by $\mathcal{F}_c=\{f_c\}_{f=1}^F$, where $f_c$ represents the $f^{\text{th}}$ content of class $c$.
	Moreover, let the content size be $S$ bits.
\end{defn}

\vspace{-0.1in}
\begin{defn}[Content Features]
	Let the content in the $c^{\text{th}}$ class have $G_c \in \mathbb{Z}^{+}$ features.
	Denote the feature set of content $f_c$ by $\mathcal{G}_{f_c} = \{g_{f_c}\}_{g=1}^{G_c}$.
	Note that the content features are essentially the descriptive attributes of the content. For example, it can be the genre/type of the content, names of the directors, actors, actresses, geolocation information, timestamp, etc.
\end{defn}
\vspace{-0.1in}
\begin{defn}[Content Library]
	The content library is comprised of all contents from all classes. 
	Let $\mathcal{F} = \bigcup_{c=1}^C \mathcal{F}_c$ be the content library.
\end{defn}

\vspace{-0.1in}
\begin{defn}[Duration of Interest (DoI)]
	The period for which the content library remains fixed is denoted as the duration of interest (DoI).
	Denote this period by $\Upsilon$.
\end{defn}
\vspace{-0.05in}
The content library $\mathcal{F}$ is fixed. 
While the CVs may request contents in each slot $t$, the edge servers can update its cached content only in each $t=n\Upsilon$ time slot, where $n \in \mathbb{Z}^{+}$ defines the cache (re)placement counter. 
Note that this assumption is made as it may not be practical to update the cache storage in each slot $t$ due to hardware limitations. 
Moreover, we assume that the content update/refresh process is independent of the content request arrival process.
To this end, we focus on request arrival modeling from the CVs, followed by modeling their preferences, i.e., their requested $f_c \in \mathcal{F}$ in slot $t$.

\vspace{-0.1in}
\subsection{User Request/Traffic Modeling}
\noindent
We assume that the CVs can make content requests in each slot $t$ following \textit{Bernoulli} distribution\footnote{Similar access modeling was also used in existing works \cite{8374824,8357917}.}.
At slot $t$, let $\Theta_u^t$ denote a \textit{Bernoulli} random variable - with success probability $p_u$, that defines whether $u$ makes a content request or not.
The total number of requests during a DoI $\Upsilon$ for an individual CV $u$ follows \textit{Binomial} distribution.
Besides, at slot $t$, the total number of requests from all CVs, i.e., $\Psi_t = \sum_{u=1}^{U} \Theta_u^t$, follows a Poisson binomial distribution \cite{wang1993number} with probabilities $\mathbf{p}=\{p_u\}_{u \in \mathcal{U}}$.
Moreover, the probability of the distribution of $\Psi_t$ can be bounded using Proposition \ref{Prop_Chernoff_Bound}.

\vspace{-0.05in}
\begin{prop} \label{Prop_Chernoff_Bound}
	Let $\bar{\mu} = \mathbb{E}[\Psi_t] = \sum_{u=1}^U p_u$ and $\bar{p} = (1/U)\sum_{u=1}^U p_u$ be the average success probability. 
	Then, at slot $t$, the probability that the distribution of the total number of requests of the CVs gets larger than some $\xi=\bar{\mu}+\delta$ and $0 < \delta < U-\bar{\mu}$, is bounded above as follows:
	\begin{equation}
		\label{Chernoff_bound}
		\begin{aligned}
			\mathrm{Pr}\left\{\Psi_t \geq \xi \right\} & \leq \exp\left[-UD_{\bar{p}} \left(\chi\right)\right],
		\end{aligned}
	\end{equation}
	where $D_{\bar{p}} \left( \chi \right) = \chi \ln \left(\chi/\bar{p}\right) + (1-\chi)\ln\left((1-\chi)/(1-\bar{p})\right)$ is the relative entropy of $\chi$ to that of $\bar{p}$ and $\chi = \xi/U$. 
\end{prop}

\begin{proof}
	Please see Appendix \ref{Appendix_A}.
\end{proof}

\subsection{Individual User Preference Modeling}
\label{Content_Request_Modeling} 
\noindent
Given that a CV makes a request, we now focus on the \emph{which} question, i.e., given that $\Theta_u^t=1$, which content shall that CV request at that slot? 
Let us express a particular content $f_{c}$ requested by CV $u$ during slot $t$ by 
\begin{equation}
	\label{Req_Indicator}
	\begin{aligned}
		\mathrm{I}^{f_c}_{u}(t) = \begin{cases}
			1, & \text{if $\Theta_u^t=1$ and $u$ request $f_c$},\\
			0, & \text{otherwise},
		\end{cases}
	\end{aligned}
\end{equation}
Unlike legacy modeling\footnote{Legacy model assumes that content requests follow Zipf distribution \cite{huang2021delay,nan2021delay,9129007}, which does not capture the personal preferences of the users\cite{Pervej_UPLEdge}.}, we consider that a CV's choice depends both on its personal preference and global popularity.
Depending on the operational needs, a CV may prefer to consume a specific content related to its operation. Besides, it may also prefer consuming a content very similar to the one that it previously consumed.
Moreover, it may also get influenced by the popularity of the contents.
For example, at slot $t$, a CV may request operation-related content $f_c$ specific to that particular timestamp. At slot $t+1$, it may then need another operation-related content $f_c^{'}$ that is very similar to $f_c$. 
Similarly, for purely entertainment-related content, the request can get influenced by the user's previous experience. The user may also choose to consume the most popular content at that time.
As such, we model the user's content request as the \emph{exploitation-exploration} tradeoff between personal preference and global popularity of the contents. 
We present this by the $\epsilon_u$-policy, i.e., a CV exploits with probability $\epsilon_u$ and explores with probability $(1-\epsilon_u)$.

\subsubsection{Content Selection during Exploitation}
In this case, the CV exploits its preferred contents from the same class it previously consumed a content. 
Given that CV has requested $f_c$ in slot $t$, it will request the most similar content to $f_c$ in class $c$ with probability $\epsilon_u$ if $\Theta_u^{t+1}=1$.
Note that similarity between $f_c$ and $f_{c}^{'}$, where $f_c \neq f^{'}_c$ is calculated as
\begin{equation}
\label{Cosine_Similarity}
	\begin{aligned}
		\Omega_{f_c}^{f^{'}_c} = \big(\sum\nolimits_{g \in \mathcal{G}_{f_c}} g_{f_c} ~ g_{f^{'}_c} \big) / \Big( \sqrt{\sum\nolimits_{g \in \mathcal{G}_{f_c}} g_{f_c}^2} \sqrt{\sum\nolimits_{g \in \mathcal{G}_{f^{'}_c}} g_{f^{'}_c}^2} \Big).
	\end{aligned}
\end{equation}

\subsubsection{Content Selection during Exploration} 
Given that the CV requested content from class $c$ previously, it will explore new content from a different class $c' \neq c$.
Denote CV $u$'s class selection probability by $p_{c}^u$, which follows a \textit{categorical} distribution.
Once the CV chooses the new content class $c$, it randomly selects a content from this class based on global popularity.
Denote the global popularity of contents in class $c$ by $\mathbf{p}_c^f=\{p_c^{f_c}\}_{f_c\in \mathcal{F}_c}$, where $p_c^{f_c}$ is the popularity of content $f_c$.

Note that our design boils down to a solely popularity-based model \cite{9552606} when there is only a single content class and $\epsilon_u=0$.

\vspace{-0.1in}
\subsection{Content (Re)placement in the Cache}
\noindent
Recall that only the edge server has limited cache storage in the proposed VEN.
Let the cache storage of the edge server be $\Lambda$.
Denote the cache placement indicator by the binary indicator function $\mathrm{I}_{f_c}(n)$ during cache placement counter $n$.
The edge server obeys the following rules for cache placement:
\begin{minipage}{0.25\textwidth}
\begin{equation}
    S \cdot \sum\nolimits_{c=1}^C \sum\nolimits_{f=1}^F  \mathrm{I}_{f_c}(n) \leq \Lambda, \eqz \label{cacheConst1} 
\end{equation}    
\end{minipage} \hspace{0.01in}
\begin{minipage}{0.22\textwidth}
    \begin{equation}
    S \cdot \sum\nolimits_{f=1}^F \mathrm{I}_{f_c}(n) = \Lambda^c \label{cacheConst2}\eqz, \eqz 
\end{equation}    
\end{minipage}
where $\Lambda^c$ is the total storage taken by the cached content from class $c \in \mathcal{C}$.
Moreover, (\ref{cacheConst1}) ensures that the size of the total cached contents must not be larger than the storage capacity.
At each $t = n \Upsilon$, the software-defined controller pushes the updated contents into the cache storage. 
These contents remain at the edge servers' cache storage till the next DoI update.

\vspace{-0.1in}
\section{Delay Modeling and Problem Formulation}
\label{I2V_Com_Delay_Model}

\vspace{-0.05in}
\subsection{Content Delivery Delay Modeling}	
\vspace{-0.05in}
\noindent 
For higher automation (and uninterrupted entertainment of the onboard users), the CVs may need to continuously access diverse contents within a tolerable delay to avoid fatalities (and quality of experiences (QoEs)). 
This motivates us to introduce a hard deadline requirement for the edge server to deliver the requested contents. 
We consider that content requests arrive continuously following $\Theta_u^t$s.
Each CV can make at most a single content request based on its preference if $\Theta_u^t=1$. 
The requester has an associated hard deadline requirement $d_{f}^{max}$, within which it needs the entire payload.
The edge server, on the other hand, can have a shorter deadline, denoted by $\hat{\mathrm{d}}_{u}^{f_c} (t)$, associated with the requests as it replaces the content at the end of each DoI. 
Formally, the deadline for the edge server to fully offload a requested content is  
\begin{equation}
\label{deadline}
    \hat{\mathrm{d}}_{u}^{f_c} (t) = \text{min} \{d_{f}^{max}, (n+1) \Upsilon - t\}, \quad \forall t \in [n\Upsilon, (n+1)\Upsilon].
\end{equation}
Essentially, (\ref{deadline}) ensures that the edge server cannot exceed the minimum of the maximum allowable delay threshold $d_f^{max}$ and remaining time till the next cache replacement slot $(n+1)\Upsilon$.

To that end, we calculate the associated delay of delivering the requested contents from all CVs $u \in \mathcal{U}$. 
This delay depends on whether the requested content has been prefetched and the underlying wireless communication infrastructure.
Particularly, for the cache miss event\footnote{When the edge server needs to serve content request $\mathrm{I}^{f_c}_u(t)$ but it does not have content $f_c$ in its local cache storage is known as the cache miss.}, the requested content is extracted from the cloud. 
This extraction causes an additional delay and involves the upper layers of the network\footnote{We assume that each miss event needs to be handled by the upper layers. 
}. 
Denote the delay for extracting content $f_c$, requested by CV $u$ during slot $t$, from the cloud by $d_{u,f_c}^{m,t}$. 
Moreover, we consider two more additional delays. 
The first one is the wait time of a request before being scheduled for transmission by the edge server, given that the edge server has either prefetched the content during the cache placement slot or the upper layers have already processed the requested content from the cloud.
The second one is the transmission delay. 
Denote these two delays, i.e., wait time and transmission delays, for $\mathrm{I}_u^{f_c}(t)$ by $d_{u,f_c}^{q,t}$ and $d_{u,f_c}^{s,t}$, respectively.
Therefore, the total delay of delivering the entire content is calculated as
\begin{equation}
\label{totalDelay}
    d^{f_c}_{u} \left(t\right) = \left[1 - \mathrm{I}_{f_c}(n)\right] d^{m, t}_{u, f_c} + d^{q,t}_{u, f_c} + d^{s, t}_{u, f_c}.
\end{equation}
Thus, the average content delivery delay for all CVs is 
\begin{equation}
\label{Average_Delay}
    \bar{\mathrm{d}}(t) = (1/U) \sum\nolimits_{u=1}^U \sum\nolimits_{c=1}^C \sum\nolimits_{f=1}^{F} \mathrm{I}_u^{f_c}(t) \cdot d_u^{f_c} (t).
\end{equation}

\vspace{-0.15in}
\subsection{Joint Problem Formulation}
\label{OriginalProblemFormulation}
\noindent
We aim to find joint cache placement $\mathrm{I}_{f_c}(n)$, user scheduling $\mathrm{I}_u(t)$, total $W(t)$ VC to form, the VC configuration $\mathcal{B}_{vc}^a(W(t))$, VC association $\mathrm{I}_{u}^{i,a}(t)$ and radio resource allocation for the serving APs in the selected VCs, i.e., $\mathrm{I}_{z}^{b,u}(t)$s to minimize long-term expected average content delivery delay for the CVs.
As such, we pose our joint optimization problem as
\begin{subequations}
\label{Original_Problem}
\begin{align}
	\tag{\ref{Original_Problem}} 
	\eqz \eqz & \underset{\mathrm{I}_{f_c} (n), \mathrm{I}_u(t), W(t), \mathcal{B}_{vc}^a\left(W(t)\right), \mathrm{I}_{u}^{i,a}(t), \mathrm{I}_{z}^{b,u}(t)} {\text{minimize}} \eqz \eqz \eqz \mathrm{d} = \underset{T \rightarrow \infty}{\limsup}~\mathbb{E} \left[ \frac{1}{T} \sum_{t=1}^T \bar{\mathrm{d}}(t)\right] \eqz\eqz\\
	&\label{Cons_1} \text{s. t.} \quad ~C_1:  \sum\nolimits_{c=1}^C \sum\nolimits_{f=1}^F \mathrm{I}_{u}^{f_c} (t) \leq 1, \quad \forall ~ u, t  \\
	& \qquad \label{Cons_2} \quad C_2:~ (\ref{cacheConst1}), (\ref{cacheConst2}), \quad \forall ~ n, \\
	& \qquad \label{Cons_3} \quad C_3:~ \sum\nolimits_{u=1}^U \mathrm{I}_u(t) \leq W(t), \quad \forall~t, \\ 
	& \qquad \label{Cons_4} \quad C_4:~ \sum\nolimits_{i=1}^{W(t)} \mathrm{I}_{u}^{i,a} (t) = 1 ,\quad \forall~u, \\
	& \qquad \label{Cons_5} \quad C_5:~ \sum\nolimits_{u=1}^U \mathrm{I}_{u}^{i,a} (t) = 1 ,\quad \forall~i, \\ 
	& \qquad \label{Cons_6} \quad C_6:~ \sum\nolimits_{u=1}^U \sum\nolimits_{i=1}^{W(t)} \mathrm{I}_{u}^{i,a} (t) = W(t), \\
	& \qquad \label{Cons_7} \quad C_7: ~ W(t) = \text{min}\left\{\sum\nolimits_{u=1}^U \mathrm{I}_u(t), \mathrm{W}_{\text{max}}\right\},\\ 
	& \qquad \label{Cons_8} \quad C_8:~ \sum\nolimits_{b=1}^B \mathrm{I}_z^{b,u} (t) = 1 ,\quad \forall~z,u ,\\
	& \qquad \label{Cons_9} \quad C_9:~ \sum\nolimits_{z=1}^Z \mathrm{I}_z^{b,u} (t) = 1 ,\quad \forall~b,u ,\\
	& \qquad \label{Cons_10} \quad C_{10}:~ \sum\nolimits_{u=1}^U \mathrm{I}_z^{b,u} (t) = 1 ,\quad \forall~z,b, \\
	& \qquad \label{Cons_11} \quad C_{11}:~ \sum\nolimits_{z=1}^Z \sum\nolimits_{b=1}^B \sum\nolimits_{u=1}^U \mathrm{I}_z^{b,u} (t) = Z , \\
	& \qquad \label{Cons_12} \quad C_{12}:~ d_u^{f_c}(t) \leq \hat{\mathrm{d}}_{u}^{f_c} (t), \\
	& \qquad \label{Cons_13} \quad C_{13}:~ \mathrm{I}_{f_c}(n) , \mathrm{I}_u(t), \mathrm{I}_{u}^{i,a}(t), \mathrm{I}_{z}^{b,u}(t) \in \{0,1\}\!,\eqz\eqz\eqz\eqz\eqz  
\end{align}
\end{subequations}
where $\mathrm{W}_{\text{max}}$ is the maximum allowable number of VCs in the system.
Constraint $C_1$ ensures that each CV can request at most one content in each slot $t$. The constraints (\ref{cacheConst1}) and (\ref{cacheConst2}) in $C_2$ are due to physical storage limitations. 
Constraint $C_3$ in (\ref{Cons_3}) restricts the total number of scheduled CVs to at max the total number of created VCs $W(t)$.
Constraints (\ref{Cons_4}), (\ref{Cons_5}) and (\ref{Cons_6}) make sure that each CV can get at max one VC, each VC is assigned to at max one CV and summation of all assigned VCs is at max the total number of available VCs, respectively.
Besides, constraints $C_7$ in (\ref{Cons_7}) restricts the total VCs $W(t)$ to be at max the minimum of the total scheduled CVs and $\mathrm{W}_{\text{max}}$.
Furthermore, constraints (\ref{Cons_8}), (\ref{Cons_9}) and (\ref{Cons_10}) ensure that each AP can get at max one pRB, each pRB is assigned to at max one AP and each CV gets non-overlapping resources, respectively.
Constraint $C_{11}$ in (\ref{Cons_11}) ensures that all available radio resources are utilized.
$C_{12}$ is introduced to satisfy the entire payload delivery delay of a requested content to be within the edge server's hard deadline $\hat{\mathrm{d}}_u^{f_c}(t)$.
Finally, constraints in (\ref{Cons_13}) are the feasibility space.
\vspace{-0.05in}
\begin{rem}{(Intuitions behind the constraints)}
Constraint $C_1$ incorporates CV's content request, while $C_2$ is for the cache placement at the edge server. 
Constraint $C_3$ is for CV scheduling.
Constraints $C_4$ - $C_7$ are for the user-centric RAT's VC formation and associations.
Besides, constraints $C_8$ - $C_{11}$ are for radio resource allocation.
Moreover, $C_{12}$ is introduced to satisfy the hard deadline for delivering the CVs' requested contents, which holds if $ \sum\nolimits_{\bar{t}=t}^{\hat{\mathrm{d}}_u^{f_c}(t)} \mathrm{I}_u(\bar{t}) \cdot \kappa \cdot R_u(\bar{t}) \geq S$, where $\kappa$ is the transmission time interval (TTI).
\end{rem}

\vspace{-0.1in}
\begin{rem}
The total delay associated with each content request, calculated in (\ref{totalDelay}), depends on both cache placement and the RAT.
More specifically, an efficient cache placement solution can minimize cache miss events, i.e., minimize $d_{u,f_c}^{m,t}$s.
On the other hand, $d_{u,f_c}^{q,t}$ and $d_{u,f_c}^{s,t}$ depend on the CV scheduling, and total VC $W(t)$, VC configuration $\mathcal{B}_{vc}^a(W(t))$, CV-VC association $\mathrm{I}_{u}^{i,a}(t)$ and radio resource allocation $\mathrm{I}_{z}^{b,u}(t)$.
\end{rem}
\vspace{-0.05in}

Note that the optimization problem in (\ref{Original_Problem}) is an average Markov decision process (MDP) over an infinite time horizon with different combinatorial optimization variables.
Recall that the CSI varies in each slot $t$.
Besides, in slot $t$, neither the CV's to-be-requested contents nor the CSI in the future time slots are known beforehand.
As such, without knowing these details, the optimal decision variables may not be known.
In the subsequent section, we will prove that even the reduced problems of this complex joint optimization are NP-hard.
Moreover, the decision variables are different in different time slots. 
As such, we decompose the original problem into two sub-problems. 
The first sub-problem transforms the cache placement problem, which will be solved using a learning solution.
The second sub-problem introduces a joint CV scheduling, total $W(t)$ VC formation, association and resource allocation optimization problem for minimizing the average content delivery delay, given that the edge server knows the cache placement decisions.
The learning solution for the cache placement depends on the following preliminaries of DRL.

\vspace{-0.15in}
\subsection{Preliminary of Deep Reinforcement Learning}
\label{DRLPrem}
\noindent
An MDP contains a set of states $\mathcal{X}=\{x\}_{x=1}^{|\mathcal{X}|}$, a set of possible actions $\mathcal{M}=\{m\}_{m=1}^{|\mathcal{M}|}$, a transition probability $P_{tt'}(m)$ from the current state $x_t \in \mathcal{X}$ to the next state $x_{t'} \in \mathcal{X}$ when action $m$ is taken, and an immediate reward $R_t(m)$ for this state transition \cite{sutton2018reinforcement}.
RL perceives the best way of choosing actions in an unknown environment through repeated observations and is widely used for solving MDP. 
The RL agent learns policy $\pi: \mathcal{M} \times \mathcal{X} \rightarrow [0,1]$, where $\pi(x_t, m) = \mathrm{Pr}\{m|x_t\}$ denotes the probability of taking action $m$ given the agent is at state $x_t$.
Following $\pi$, given the agent is at state $x_t$, the expected return from that state onward denotes how good it is to be at that state and is measured by the following state-value function:
\begin{equation}
\label{state_val_func}
V_{\pi}(x_t) = \mathbb{E}[R|x_t,\pi] = \mathbb{E} 
\left[\sum\nolimits_{t'= t}^{\mathrm{T}_{\text{end}}} \gamma^{t'-t} R_{t'}(m)|x_t, \pi \right],
\end{equation}
where $\gamma \in [0,1]$ is the discount factor, $\mathrm{T}_{\text{end}}$ is the time step at which the episode ends, and $R_{t'}(m)$ is the reward at step $t'$.
Moreover, the quality of an action taken at a state is ascertained by the following action-value function \cite{sutton2018reinforcement}:
\begin{equation}
    Q(x_t,m) = R_t(m) + \gamma \sum\nolimits_{x_{t'}\in \mathcal{X}} P_{tt'}(m) V_{\pi}(x_{t'}).
\end{equation}
The agent's goal is to find optimal policy $\pi^{*}$ to maximize
\begin{subequations}
\label{Q_pi}
\begin{align}
	Q^{*}(x_t,m) &= R_t(m) + 
	\gamma \sum\nolimits_{x_{t'}\in \mathcal{X}} P_{tt'}(m) V_{\pi^{*}}(x_{t'}),
\end{align}
\end{subequations}
where $V_{\pi^{*}}(x_{t'})=\underset{\breve{m} \in \mathcal{M}}{\mathrm{max}} ~ Q^{*}(x_{t'},\breve{m})$.
This $Q(x_t, m)$ value is updated as \cite{watkins1992q}
\begin{equation}
\label{Q_Value}
\begin{aligned}
	Q \left(x_t, m\right)  & \leftarrow (1-\alpha) Q \left(x_t, m\right) + \alpha \bar{y}_t,
\end{aligned}
\end{equation}
where $\alpha$ is the learning rate and $\bar{y}_t=R_{t}(m) + \gamma \underset{\breve{m} \in \mathcal{M}}{\text{ max }}Q \left(x_{t'}, \breve{m} \right))$ is commonly known as the temporal target.
Usually, a deep neural network (DNN), parameterized by its weight $\pmb{\theta}$, is used to approximate $Q^{*}(x,m)\approx Q(x,m;\pmb{\theta})$, which is known as the so-called DRL \cite{mnih2015human}. The agent is trained by randomly sampling $\mathrm{S}_{b}$ batches from a memory buffer $\mathcal{D}$, which sores of the agent's experiences $\{x_t, m, R_t, x_{t'}\}$, and performing stochastic gradient descent (SGD) to minimize the following loss function \cite{mnih2015human}:
\begin{equation}
\label{train_loss}
    \mathrm{L}(\pmb{\theta}) = [
    \bar{y}_t(\pmb{\theta})- Q(x_t, m; \pmb{\theta}) ]^2,
\end{equation} 
where $\bar{y}_t(\pmb{\theta})=R_{t}(m) + \gamma \underset{\breve{m} \in \mathcal{M}} {\text{ max }}Q \left(x_{t'}, \breve{m};\pmb{\theta} \right)$.
While the same DNN $\pmb{\theta}$ can be used to predict both $Q(x_t,m; \pmb{\theta})$ and the target $\bar{y}_t(\pmb{\theta})$, to increase learning stability, a separate target DNN, parameterized by $\pmb{\theta^{-}}$, is used to predict $\bar{y}_{t}(\pmb{\theta}^{-})$ \cite{mnih2015human}.

Here we emphasize that since the $Q$ value functions are estimated using the DNN $\pmb{\theta}$, i.e., $Q^{*}(x,m)\approx Q(x,m;\pmb{\theta})$, unlike classical tabular Q-learning, the DRL solution may not be optimal \cite{sutton2018reinforcement}.
To that end, we first introduce our problem transformation in the next section, followed by more pertinent information on a DRL-based solution in Section \ref{problemSolutionSec}.

\vspace{-0.1in}
\section{Problem Transformations}
\label{probTransformationSec}
\noindent
Since the original problem is hard to solve and the decision variables are not the same in different time slots, we decompose the original problem by first devising a learning solution for cache placement policy (CPP) for the cache placement slot $t=n \Upsilon$.  
Then, we use this learned CPP to re-design the delay minimization problem from the RAT perspective. Intuitively, given that the best CPP for the slot $t=n \Upsilon$ is known, in order to ensure minimized content delivery delay, one should optimize the RAT parameters jointly.

\vspace{-0.15in}
\subsection{Cache Placement Policy (CPP) Optimization Sub-Problem}
\noindent
We want to learn the CPP $\pi_{ca}$ that provides the optimal cache placement decision $\mathrm{I}_{f_c}(n)$ in the cache placement time slots $t=n\Upsilon, ~\forall n$.
We have total $\mathrm{m}_{ca}=\prod_{c=1}^{C} \mathrm{m}_{ca}^c = \prod_{c=1}^{C} \binom{F}{\Lambda^c}$ ways for content placement as $\Lambda^c$ and $\Lambda$ are of the unit of content size $S$ based on constraints (\ref{cacheConst1}) and (\ref{cacheConst2}).
Moreover, the CPP $\pi_{ca}$ is a mapping between the system state $x_{ca}^n$ and an action $m_{ca}$ in the joint action space with $\mathrm{m}_{ca}(n)$ possible actions.
To this end, let us define a cache hit event by
\begin{equation}
\mathbbm{1}_{\mathrm{I}_u^{f_c}} (t) = \begin{cases}
	1, & \text{if } \mathrm{I}_u^{f_c}(t)=1 \text{ and } \mathrm{I}_{f_c}(n)=1, \\
	0, & \text{otherwise}.
\end{cases}.
\end{equation}
Besides, the total cache hit at the edge server is calculated as the summation of the locally served requests and is calculated as $h(t) = \sum_{u=1}^U \mathbbm{1}_{\mathrm{I}_u^{f_c}} (t)$.
Thus, we calculate the CHR as
\begin{equation}
\label{cache_hit_ratio}
    \mathrm{CHR}(t) = h(t) / \left( \sum\nolimits_{u=1}^{U} \mathrm{I}_u^{f_c}(t)\right). 
\end{equation}
Next, we devise the CPP of the edge server that ensures a long-term CHR while satisfying the cache storage constraints.
Formally, we pose the optimization problem as follows:
\begin{subequations}
\label{Cache_Hit_Max_Problem}
\begin{align}
	\tag{\ref{Cache_Hit_Max_Problem}} & \underset{ \pi_{ca} }  {\text{maximize}} &&  \mathrm{CHR}(\pi_{ca}) = \underset{T \rightarrow \infty}{\lim}~\mathbb{E} \big[ (1/T) \sum\nolimits_{t=1}^T \mathrm{CHR} (t) \big], \\
	\label{CH_P_Cons_1} & \quad \text{s. t.}  && C_1, C_2, ~ \mathrm{I}_{f_c}(n) \in \{0,1\},
\end{align}
\end{subequations}
where $C_1$ and $C_2$ are introduced in (\ref{Original_Problem}).

\begin{thm}
\label{Thm_CHR_NP_Hardness}
The CHR maximization problem (\ref{Cache_Hit_Max_Problem}) is NP-hard.
\end{thm}
\begin{proof}
Please see Appendix \ref{Appendix_CHR_NP_Hardness}.
\end{proof}

\vspace{-0.35in}
\subsection{Joint Optimization Problem for the User-Centric RAT}
\noindent
Note that as the delay of extracting a content from the cloud during a cache miss event is fixed, the first term in (\ref{totalDelay}) will be minimized if the CPP $\pi_{ca}$ ensures maximized CHR.
In this sub-problem, we focus on the other two delays $d_{u,f_c}^{q,t}$ and $d_{u,f_c}^{s,t}$ in (\ref{totalDelay}) by jointly optimizing scheduling, VC formation, VC association and radio resource allocation of the proposed user-centric RAT solution assuming the cache placement is known at the edge server. 
Therefore, we pose the following modified content delivery delay minimization problem. 
\begin{subequations}
\label{Delay_Minimization}
\begin{align}
	\tag{\ref{Delay_Minimization}} 
	\!\!\!\!\!\!\! & \underset{\!\!\!\!\!\!\!\!\!\!\!\! \mathrm{I}_u(t), W(t), \mathcal{B}_{vc}^a\left(W(t)\right), \mathrm{I}_{u}^{i,a}(t), \mathrm{I}_{z}^{b,u}(t)} {\text{minimize}} \!\!\! \mathrm{d} = \underset{T \rightarrow \infty}{\limsup}~\mathbb{E} \big[ (1/T) \sum\nolimits_{t=1}^T \bar{\mathrm{d}}(t)\big] \!\!\!\!\! \\
	& ~ \label{Cons1_1} \qquad \quad \text{s. t. } \quad  C_3, C_4, C_5, C_6, C_7,C_8,C_9,C_{10},C_{11}, C_{12}, \\
	& \label{Cons1_2} \qquad \qquad \qquad ~ \mathrm{I}_u(t), \mathrm{I}_{u}^{i,a}(t), \mathrm{I}_{z}^{b,u}(t) \in \{0,1\}\!,\eqz\eqz\eqz\eqz\eqz  
\end{align}
\end{subequations}
where the constraints in (\ref{Cons1_1}) and (\ref{Cons1_2}) are taken for the same reasons as in the original problem in (\ref{Original_Problem}).

Sub-problem (\ref{Delay_Minimization}) contains combinatorial optimization variables and, thus, is NP-hard.
An exhaustive search for optimal parameters is also infeasible due to the large search space as well as sequential dependencies for the deadline constraints in $C_{12}$.
Besides, as each content request arrives with a deadline constraint and wireless links vary in each slot, we consider that the edge server adopts priority-based scheduling.  
Intuitively, given the fact that the edge server does not know the transmission delay $d_{u,f_c}^{s,t}$ due to channel uncertainty and it needs to satisfy constraint $C_{12}$ for all $\mathrm{I}_u^{f_c}(t)$s, it should schedule the CVs with earliest-deadline-first (EDF)\footnote{Similar scheduling is also widely used in literature \cite{8624268,9348504}.} followed by optimal VC formation, association and radio resource allocation.
Note that EDF is widely used for scheduling in real-time operating systems \cite{buttazzo2011hard}. 
If EDF cannot guarantee zero deadline violation for the tasks, no other algorithm can \cite{9348504}.
In our case, scheduling also depends on the availability of the requested content at the edge server.
In cache miss event, the edge server must wait for $d_{u,f_c}^{m,t}$ so that the upper layers can extract the content from the cloud.

Upon receiving a content request $\mathrm{I}_{u}^{f_c}(t)$, the edge server checks $\mathrm{I}_{f_c}(n)$. 
If $\mathrm{I}_{f_c}(n)=0$, the request is forwarded to the upper layers. 
The upper layer initiates the extraction process from the cloud.
At each slot $t$, before making the scheduling and VC formation decisions, the edge server considers previous $\mathcal{T}_{\text{SoI}}^t$ slots information. 
These $\mathcal{T}_{\text{SoI}}^t$ slots are termed as our slots of interest (SoI) and are calculated in (\ref{SoI}).
\begin{equation}
\label{SoI}
    \mathcal{T}_{\text{SoI}}^t = \left\{\text{min} \{0, t-d_{f}^{max}+\zeta\} \right\}_{\zeta=1}^{d_f^{max}}.
\end{equation}
This SoI captures the previous slots that may still have undelivered payloads with some remaining time to the deadlines at the current slot $t$.
Denote the remaining time to the deadline and payload for  $\mathrm{I}_{u}^{f_c}(t-d_f^{max}+\zeta)$ in current slot $t$ by $\mathrm{T}_{u,\text{rem}}^{t-d_f^{max}+\zeta}$ and $\mathrm{P}_{u,\text{rem}}^{t-d_f^{max}+\zeta}$, respectively.
Particularly, for all $\mathrm{I}_u^{f_c}(t-d_f^{max} + \zeta)$, the edge server first checks whether the content is available at the edge server's local cache storage or by the upper layers. 
If it is available, the edge server calculates the remaining time to the deadlines and payloads for the requests in all slots of $\mathcal{T}_{\text{SoI}}^{t}$.
The edge server finds a set of candidate requester CVs $\mathcal{U}^t_{\text{val}} \subseteq \mathcal{U}$, their minimum remaining time to the deadline set $\mathcal{T}_{\text{rem}}^t$ and corresponding left-over payload set $\mathcal{P}_{\text{rem}}^t$.
This procedure is summarized in Algorithm \ref{validCVSetExtraction}.
Note that the time complexity of Algorithm \ref{validCVSetExtraction} is $\mathcal{O}(4U|\mathcal{T}_{\mathrm{SoI}}^t| + 3U + 5)$.

After extracting the valid CV set $\mathcal{U}_{\text{val}}^t$, the edge server can formulate total $W(t)$ VCs based on the following equation:
\begin{equation}
\label{VC_To_form}
	W(t) = \text{min} \{ \left\vert \mathcal{U}_{\text{val}}^t \right\vert, \mathrm{W}_\text{max} \},
\end{equation}
where $\left \vert \mathcal{U}_{\text{val}}^t \right \vert$ is the cardinality of the set $\mathcal{U}_{\text{val}}^t$.
This essentially means that the server creates the minimum of the total valid CVs in the set $\mathcal{U}_{\mathrm{val}}^t$ and the maximum allowable number of VCs.
Besides, the edge server calculates the priorities of the valid CVs set based on their remaining time to the deadlines using the following equation:
\begin{equation}
\label{deadline_weight_in_rt}
    \phi_u(t) = \bar{\phi}_u(t) / \left(\sum\nolimits_{u \in\mathcal{U}_{\text{val}}^t} \bar{\phi}_u(t)\right), 
\end{equation}
where $\bar{\phi}_u(t) = \left(\sum\nolimits_{u \in \mathcal{U}_{\text{val}}^t} \mathcal{T}_{\text{val}}^t[u] \right) / \mathcal{T}_{\text{val}}^t[u] $.
Note that (\ref{deadline_weight_in_rt}) sets the highest priority to the CV that has the least remaining time to the deadline.
The edge server then picks the top-$W(t)$ CVs for scheduling based on the priorities $\phi_u(t)$s.
Denote the scheduled CV set during slot $t$ by $\mathcal{U}_{\text{sch}}^t \subseteq \mathcal{U}_{\text{val}}^t$.
Given that the edge server makes scheduling decisions based on top-$W(t)$ priorities of (\ref{deadline_weight_in_rt}), to satisfy the hard deadline constraint in $C_{12}$, we aim to maximize a WSR, which is calculated as
\begin{equation}
\label{wsr_cvs}
    \breve{R}(t) =  \sum\nolimits_{u \in \mathcal{U}_{\text{sch}}^t } \mathrm{I}_u(t) \cdot R_u(t) \cdot \phi_u(t),
\end{equation}
where the weights are set based on the CV's priority $\phi_u(t)$. 
Again, the intuition for this is that with the underlying RAT solution, due to channel uncertainty, the edge server expects to satisfy constraint $C_{12}$ by prioritizing the CVs based on (\ref{deadline_weight_in_rt}) and follow optimal VC configuration, their association and radio resource allocation.
As such, we pose the following WSR maximization problem for the edge server: 
\begin{subequations}
\label{Sum_Rate_Maximization}
\begin{align}
	\tag{\ref{Sum_Rate_Maximization}} \underset{\mathcal{B}_{vc}^a\left(W(t)\right), \mathrm{I}_{u}^{i,a}(t), \mathrm{I}_{z}^{b,u}(t)} {\text{maximize}} &  \qquad  \breve{R} (t), \\
	\label{Sum_Rate_Max_Cons_1} \text{subject to} ~\quad & \qquad C_4, C_5, C_6, C_8, C_9, C_{10}, C_{11},\\
	\label{Sum_Rate_Max_Cons_2} \quad & \qquad \mathrm{I}_{u}^{i,a}(t)\in \{0,1\}, \mathrm{I}_{z}^{b,u}(t) \in \{0,1\}, 
\end{align}
\end{subequations}
where the constraints in (\ref{Sum_Rate_Max_Cons_1}) and (\ref{Sum_Rate_Max_Cons_2}) are taken for the same reasons as in the original problem in (\ref{Original_Problem}).
\vspace{-0.1in}
\begin{rem}
	The edge server finds $W(t)$ VCs and $\mathrm{I}_u(t)$s using (\ref{VC_To_form}) and (\ref{deadline_weight_in_rt}), respectively.
    Given that the contents are placed following $\pi_{ca}$ during slot $t=n\Upsilon$, and the edge server knows $W(t)$ and $\mathrm{I}_u(t)$s, the joint optimization problem in (\ref{Delay_Minimization}) is simplified to a joint VC configuration, CV-VC association and radio resource allocation problem in (\ref{Sum_Rate_Maximization}).
\end{rem}

\vspace{-0.15in}
\section{Problem Solution}
\vspace{-0.05in}
\label{problemSolutionSec}
\noindent
The edge server uses a DRL agent to solve the transformed CHR maximization problem (\ref{Cache_Hit_Max_Problem}). 
Since the CVs request contents based on the preference-popularity tradeoff and their future demands are unknown to the edge server, DRL is adopted as a sub-optimal learning solution for (\ref{Cache_Hit_Max_Problem}).
Moreover, we optimally solve the joint optimization problem (\ref{Sum_Rate_Maximization}).
In order to do so, first, we leverage graph theory to find optimal pRB allocation based on a given VC configuration $\mathcal{B}_{vc}^a(W(t)) \in \mathcal{B}_{vc}(W(t))$. Then we perform a simple linear search to find the best VC configuration $\mathcal{B}_{vc}^{a^{*}}(W(t))$.

\vspace{-0.1in}
\begin{algorithm} [t!]
\fontsize{8}{7}\selectfont
\SetAlgoLined
\KwIn{$\mathcal{T}_{\text{SoI}}^t$, $\{\mathrm{I}_{u}^{f_c}(k)\}_{k\in\mathcal{T}_{\text{SoI}}^t}$, $\{\mathrm{T}_{u,\text{rem}}^{k}\}_{k \in \mathcal{T}_{\text{SoI}}^t}$, $\{\mathrm{P}_{u,\text{rem}}^{k}\}_{k \in \mathcal{T}_{\text{SoI}}^t}$ } 
Initiate empty valid CV set $\mathcal{U}_{\text{val}}^t = [\cdot]$, valid minimum time to the remaining deadline set $\mathcal{T}_{\text{rem}}^{t} = [\cdot]$ and valid remaining payload set $\mathcal{P}_{\text{rem}}^{t} = [\cdot]$ \;
Initiate an initial deadline matrix $\pmb{\mathrm{T}}_{rem}^t \gets Ones(U \times d_f^{max}) \times 100 $ \;
Initiate an initial payload matrix $\pmb{\mathrm{P}}_{rem}^t \gets Zeros(U \times d_f^{max})$ \;
\For{all $k$ slots in SoI set $\mathcal{T}_{\text{SoI}}^t$}{ 
	\For{$u \in \mathcal{U}$} {
		\uIf{$\mathrm{I}_u^{f_c}(k)$ is available at the edge server in this slot $t$ $\mathrm{and}$ $\mathrm{T}_{u,\text{rem}}^k>0$ $\mathrm{and}$ $\mathrm{P}_{u,\text{rem}}^k > 0$} {
			$ \pmb{\mathrm{T}}_{rem}^t[u, k] \gets \mathrm{T}_{u,\text{rem}}^k$ \;
			$ \pmb{\mathrm{P}}_{rem}^t[u, k] \gets \mathrm{P}_{u,\text{rem}}^k$ \;
		} 
	} 
}
\For{$u \in \mathcal{U}$}{
	Find the maximum remaining payload for CV $u$ in all slots inside the SoI $\mathcal{T}_{\text{SoI}}^t$ as $P_{max}^u=\mathrm{max}\{\pmb{\mathrm{P}}_{\text{rem}}^{t} [u,:]\}$\;
	\uIf{$P_{max}^u >0$}{
		Find the minimum valid remaining time to the deadline $k_{min}^{val} = \text{min}\{\pmb{\mathrm{T}}_{rem}^t[u,:]\}$ and corresponding slot index $k_{min}^{idx}$ \;
	 	$\mathcal{U}_{\text{val}}^t.\text{append}(u)$ \;
		$\mathcal{T}_{\text{rem}}^{t}.\text{append}(k_{min}^{val})$ \;
		$\mathcal{P}_{\text{rem}}^{t} .\text{append}\left(\pmb{\mathrm{P}}_{\text{rem}}^{t} \left[u,k_{min}^{idx}\right]\right)$\;
 	}
}
\KwOut{$\mathcal{U}_{\text{val}}^t$, $\mathcal{T}_{\text{rem}}^{t}$ and $\mathcal{P}_{\text{rem}}^{t}$}
\caption{Get Eligible CV Set for Scheduling}
\label{validCVSetExtraction}
\end{algorithm}

\vspace{-0.08in}
\subsection{Learning Solution for the CPP}
\noindent
To find the CPP $\pi_{ca}$, the edge server uses some key information from the environment and learns the underlying environment dynamics.
Recall that the CVs requests are modeled by the exploration and exploitation manner. 
At the beginning of each DoI, the edge server determines top-$\Lambda^c$ popular contents in each class and also calculates top-$\Lambda^c$ similar contents for each of these popular contents as
\begin{equation*}
\label{topPopSimCont}
\mathbf{F}^{top}(n)[c,f_c] = \begin{cases}
	1, & \text{if $f_c$ is top-$\Lambda^c$ similar content of $f_c^{top}$},\\
	0, & \text{otherwise},
\end{cases},
\end{equation*}
where $f_c^{top}$ is in top-$\Lambda^c$ popular content list of class $c$.
Besides, the edge server also keeps track of the content requests coming from each CV and corresponding cache hit based on the stored content during the previous DoI. 
Let the edge server store the content-specific request from CV $u$ into a $\mathbb{R}^{C \times F}$ matrix $\mathbf{P}_{req}^{u}(n)$ during all slots of the DoI.
Similarly, let there be a matrix $\mathbf{P}_{hit}^{u}(n) \in \mathbb{R}^{C \times F}$ that captures content-specific cache hit $\mathbbm{1}_{\mathrm{I}_{u}^{f_c}}(t)$s during all $t$ within the DoI.
Furthermore, we also provide the measured popularity matrix $\mathbf{P}_f(n)$ during the current DoI based on the CVs requests in the previous DoI change interval $(n-1)$.
As such, the edge server designs \textbf{state} $x_{ca}^{n}$ as the following tuple:
\begin{equation}
\label{cacheAgentState}
x_{ca}^{n} = \{\{\mathbf{P}_{req}^{u}(n)\}_{u=1}^{U}, \{\mathbf{P}_{hit}^{u}(n)\}_{u=1}^{U}, \mathbf{F}^{top}(n), \mathbf{P}_f(n)\}.
\end{equation}
The intuition behind this state design is to provide the edge server some context on how individual CVs' preferences and global content popularity may affect the overall system reward.

At each $t=n\Upsilon$, the edge server takes a cache placement action $m_{ca}$ to prefetch the contents in its local storage.
At the end of the DoI, it gets the following \textbf{reward} $r_{ca}^n$
\begin{equation}
\label{cacheAgentReward}
    r_{ca}^n =
    (1/\Upsilon) \sum\nolimits_{\tilde{t}=t}^{(n+1)\Upsilon} r_{ca} (\tilde{t}),
\end{equation}
where $r_{ca}(\tilde{t}) = \sum_{c=1}^C \sum_{f_c=1}^{F_c} r_{ca}[c,f_c]$.
Moreover, $r_{ca}[c,f_c]$ is calculated in (\ref{cacheAgentScaledRt}), where $\delta_{pop}^{sim}$ and $\delta_{hit}$ are two hyper-parameters. 
Note that these hyper-parameters balance the cache hit for the top-$\Lambda^c$ contents and the other stored contents in the edge server's cache storage. 
Empirically, we have observed $\delta_{pop}^{sim} > \delta_{hit}$ works well.
\begin{equation}
\label{cacheAgentScaledRt}
\fontsize{8}{7}\selectfont
\small
\eqz \eqz r_{ca}[c,f_c] (\tilde{t}) = \begin{cases}
	\delta_{pop}^{sim} \cdot \sum_{u=1}^U \mathbbm{1}_{\mathrm{I}_u^{f_c}} (\tilde{t}), & \text{if $ \mathbf{F}^{top}(n)[c, f_c]=1$} \\
	& \quad \text{and $\sum_{u=1}^{U} \mathrm{I}_u^{f_c} (\tilde{t}) > 0 $}, \\
	\delta_{hit} \cdot \sum_{u=1}^U \mathbbm{1}_{\mathrm{I}_u^{f_c}} (\tilde{t}), & \text{if $ \mathbf{F}^{top}(n)[c, f_c] \neq 1$} \\
	& \quad \text{and $\sum_{u=1}^{U} \mathrm{I}_u^{f_c} (\tilde{t}) >0 $},\\
	-  \sum_{f_c=1}^{F} \sum_{u=1}^{U} \mathrm{I}_u^{f_c}(\tilde{t}), & \text{otherwise},
\end{cases}\eqz,\eqz
\end{equation}
We consider that the edge server learns the CPP $\pi_{ca}$ offline.
It uses two DNNs - $\pmb{\theta}_{ca}$ and $\pmb{\theta}_{ca}^{-}$, and learns $\pi_{ca}$ following the basic principles described in Section \ref{DRLPrem}.
Algorithm \ref{CPP_Algorithm} summarizes the CPP learning process.
While the training episode is not terminated, in line $6$, the CVs make content requests.
During the cache placement slots $t=n\Upsilon$, line $7$, the edge server observes its state $x_{ca}^n$ in line $8$.
Based on the observed state, the agent takes action $m_{ca}$ following the $\epsilon$-greedy policy \cite{sutton2018reinforcement} using $\pmb{\theta}_{ca}$ in line $9$.
During the last time slot of the current DoI, in line $11$, the environment returns the reward $r_{ca}^n$ and transits to the next state $x_{ca}^{n'}$ in line $12$.
Moreover, in line $13$, the edge server stores its experiences tuple $\left\{x_{ca}^{n}, m_{ca}, r_{ca}^n, x_{ca}^{n'} \right\}$ into its memory buffer $mem_{ca}$, which can hold $mem_{ca}^{max}$ number of samples.
In line $15$, the edge server randomly sample $S_{ca}$ batches from $mem_{ca}$ and uses the $\pmb{\theta}_{ca}$ and $\pmb{\theta}_{ca}^{-}$ to get $Q(x_{ca}^n, m_{ca};\pmb{\theta}_{ca})$ and the target value $\bar{y}_t(\pmb{\theta}^{-})$, respectively.
In line $16$, it then trains the DNN $\pmb{\theta}_{ca}$ by minimizing the loss function shown in (\ref{train_loss}) using SGD.
Moreover, after $\breve{\eta}_{ca}$ steps, the offline DNN $\pmb{\theta}_{ca}^{-}$ gets updated by $\pmb{\theta}_{ca}$ in line $20$.

\vspace{-0.1in}
\subsection{WSR Maximization} 
\noindent
Recall that once the edge server determines $W(t)$ based on (\ref{VC_To_form}), all possible VC configurations $\mathcal{B}_{vc}(W(t)) =\left\{\mathcal{B}_{vc}^a\left(W(t)\right) \right\}_{a=1}^{A_{W(t)}}$ can be generated following the VC formation rules defined in (\ref{VC_Cons1})-(\ref{VC_Cons3}).
Besides, each VC configuration $\mathcal{B}_{vc}^a\left(W(t)\right)$ has exactly $W(t)$ number of VCs.
Moreover, the edge server schedules $|\mathcal{U}_{\text{sch}}^t|=W(t)$ CVs in each slot $t$ based on the priority $\phi_{u}(t)$.
Let the $i^{\text{th}}$ CV in $\mathcal{U}_{\text{sch}}^t$ be assigned to the $i^{\text{th}}$ VC in $\mathcal{B}_{vc}^a\left(W(t)\right)$.
This assigns each CV to exactly one VC and all VCs are assigned to all scheduled CVs.
Therefore, essentially, for a selected VC configuration $\mathcal{B}_{vc}^a\left(W(t)\right)$, by assigning the VCs in the above mentioned way, the edge server can satisfy constraints $C_3$, $C_4$, $C_5$ and $C_6$.
To this end, given that the selected VC configuration $\mathcal{B}_{vc}^a\left(W(t)\right)$ and $\mathrm{I}_{u}^{i,a}(t)$ are known at the edge server, we can rewrite (\ref{Sum_Rate_Maximization}) as follows:
\begin{subequations}
\label{WSR_Max_pRB}
\begin{align}
	\tag{\ref{WSR_Max_pRB}} \underset{\mathrm{I}_{z}^{b,u}(t)} {\text{maximize}} &  \qquad  \breve{R} (t), \\
	\label{WSR_Max_pRB_Cons_1} \text{subject to}  & \qquad C_8, C_9, C_{10}, C_{11}, \mathrm{I}_{z}^{b,u}(t) \in \{0,1\}. 
\end{align}
\end{subequations}
As the CSI is perfectly known at the edge server, it can choose maximal ratio transmission to design the precoding vector $\mathbf{w}_{b}^{u,z}$.
In other words, given $\mathrm{I}_{z}^{b,u}(t)=1$, the edge server chooses $\mathbf{w}_{b}^{u,z}(t)=\mathbf{h}_{b}^{u,z}(t)/\left\Vert \mathbf{h}_{b}^{u,z}(t) \right \Vert$.
Besides, the received SNR at the CV $u$, calculated in (\ref{V2I_SINR}), is the summation over all APs of the CV's assigned VC divided by total noise power. 
As such, we can stack the weighted data rate at the CV from the APs that are in its serving VC over all pRBs into a matrix - denoted by $\mathbf{R}_t \in \mathbb{R}^{B \times Z}$ matrix. 
This weighted data rate matrix extraction process is summarized in Algorithm \ref{Extract_WSR_Matrix_per_given_Sch_CVs_VC_Selection}.
In this algorithm, we initiate a matrix of zeros of $\mathbb{R}^{B \times Z}$ in line $1$.
Recall that all VCs are assigned to the scheduled CVs and all APs are assigned to form the VCs based on the rules defined in Section \ref{comSysModel}. 
As such, for each $u \in \mathcal{U}_{sch}^t$, we get the assigned VC in line $3$.
Then, for all APs and all pRBs, we calculate the spectral efficiency in line $6$.
Moreover, we update the respective $(b,z)$ element of the $\mathbf{R}_t$ matrix in line $7$.
Note that Algorithm \ref{Extract_WSR_Matrix_per_given_Sch_CVs_VC_Selection}'s time complexity is $\mathcal{O} \big(W(t) \big[2 Z|VC_a^i| +1 \big] + 1 \big)$.

\SetInd{0.9em}{0.1em}
\begin{algorithm}[!t]
\fontsize{8}{6}\selectfont
\SetAlgoLined
\SetKwInput{KwData}{Input}
\KwData{$S_{ca}$, $Mem_{ca}$, $\breve{\eta}_{ca}$,  $\epsilon_{max}$, $\epsilon_{min}$, $\nu$, $T_{epoch}$, $\Upsilon$, $\pmb{\theta}_{ca}$, $\pmb{\theta}_{ca}^{-}$} 
Calculate $\epsilon$ decaying rate as $decay_\epsilon = \frac{\epsilon_{max} - \epsilon_{min}}{\nu \times T_{epoch}}$ \; 
\For{$e$ in $T_{epoch}$} {
	$\epsilon \gets max\{\epsilon_{min}, ~ \epsilon_{max} - \left(e \times decay_\epsilon\right) \}$ \;
	Set $t=0$, $n=0$, $done$=$\mathrm{False}$ \;
	\While{not $done$} {
		Get all CVs content requests using Section \ref{Content_Request_Modeling} \;
        \uIf {$t==0$ or  $\left((t+1) \mod \Upsilon\right)==0$} {
		    Get state $x_{ca}^{n}$ from the environment \;
		    Edge server takes cache placement action $m_{ca}$ based on observation $x_{ca}^{n}$ using its action selection policy \; 
			$n \mathrel{+}= 1$ \;
		}
        \uElseIf{$(t+1)==\left(n \Upsilon-1\right)$}  {
			Get $r_{ca}^n$, $x_{ca}^{n'}$ and $done$ flag from environment \;
			Store $\{x_{ca}^n, m_{ca}, r_{ca}^n, x_{ca}^{n'}, done\}$ into $Mem_{ca}$ \;
			\uIf(\tcp*[h]{train $\pmb{\theta}_{ca}$}){$len(Mem_{ca}) \geq S_{ca}$}{
				Uniformly sample $S_{ca}$ samples from $Mem_{ca}$ \;
				Use the sampled samples to train $\pmb{\theta}_{ca}$ \;
			}
			$x_{ca}^{n} \gets x_{ca}^{n'}$ \;
		}
		$t \mathrel{+}= 1$\;
		\uIf(\tcp*[h]{Update  $\pmb{\theta}_{ca}^{-}$}) {$t \mod \breve{\eta}_{ca}==0$} {
			$\pmb{\theta}_{ca}^{-} \gets \pmb{\theta}_{ca}$
		}
	}
}
\KwOut{$\pmb{\theta}_{ca}$}
\caption{CPP Learning Algorithm}
\label{CPP_Algorithm}
\end{algorithm}	
\SetInd{0.9em}{0.1em}
\begin{algorithm}[!t]
	\fontsize{8}{6}\selectfont
	\SetAlgoLined
	\SetKwInput{KwData}{Input}
	\KwData{$\mathcal{U}_{\text{sch}}^t$, $\mathcal{B}_{vc}^a(W(t))$, $\{\phi_u(t)\}_{u \in \mathcal{U}_{\text{sch}}^t}$, $\mathbf{H}_u^b(t)$ } 
	Initiate matrix $\mathbf{R}_t = zeros(B \times Z)$ \;
	\For{$u \in \mathcal{U}_{\text{sch}}^t$} {
		Get assigned $VC_a^i$ using $\mathrm{I}_{u}^{i,a}(t)$ \;
		\For{$b \in VC_a^i$}{ 
			\For {$z \in \mathcal{Z}$}{
				Calculate $r_t(u,b,z) = \log_2\left(1 + \frac{P_b \left\vert \mathbf{h}_{b}^{u,z}(t)^H \mathbf{w}_b^{u,z}(t) \right\vert^2}{\sigma_b^2}\right) $ received at CV $u$ from AP $b$ over pRB $z$ during slot $t$ \;
				$\mathbf{R}_t[b,z] = r_t(u,b,z) \times \phi_u(t)$ \;
			} 
			}
	}  
	\KwOut{$\mathbf{R}_t$}
	\caption{Get Weighted Data Rate Matrix}
\label{Extract_WSR_Matrix_per_given_Sch_CVs_VC_Selection}
\end{algorithm}
\begin{algorithm} [t!]
	\fontsize{8}{6}\selectfont
	\SetAlgoLined
	\SetKwInput{KwData}{Input}
	\KwIn{$W(t)$, $\mathcal{U}_{\text{sch}}^t$, $\{\phi_u(t)\}_{u \in \mathcal{U}_{\text{sch}}^t}$, $\mathbf{H}_u^b(t)$, $\mathcal{B}_{vc}(W(t))$} 
	Initiate WSR vector $\breve{\mathbf{r}}_t = zeros(A_{W(t)} )$ and empty pRB allocation set $\pmb{\mathrm{I}}_z(t) = [\cdot]$\;
	\For{$\mathcal{B}_{vc}^a(W(t)) \in \mathcal{B}_{vc}(W(t))$} {
		Get $\mathbf{R}_t$ matrix from Algorithm \ref{Extract_WSR_Matrix_per_given_Sch_CVs_VC_Selection} for this VC configuration \;
		Solve the MWBM problem using Hungarian algorithm \cite{kuhn1955hungarian} to get optimal pRB allocation set $\pmb{I}_z^{*}=\{\mathrm{I}_{z}^{b,u}(t)\}_{z=1}^Z$ and	get the optimal sum-weights $r_t^{*}$ from the optimal edges $e^{*}(b,z)$ \;
		$\breve{\mathbf{r}}_t[a]=r_t^{*}$ \;
		$\pmb{\mathrm{I}}_z(t).append(\pmb{I}_z^{*})$ \;
	}  
	Find the $\text{max}(\breve{\mathbf{r}}_t)$ and corresponding index $a^{*}$ \;
	Take best VC configuration $\mathcal{B}_{vc}^{a^{*}}(W(t))$ and corresponding optimal pRB allocation set ${\pmb{\mathrm{I}}_{z}^{b,u}}^{*}(t) = \pmb{\mathrm{I}}_z(t)[a^{*}]$\;
	\KwOut{$\mathcal{B}_{vc}^{a^{*}}(W(t))$ and ${\pmb{\mathrm{I}}_{z}^{b,u}}^{*}(t)$}
	\caption{Optimal VC Configuration and pRB Allocation}
	\label{getPRBAllocationAndVCConfig}
\end{algorithm}
\begin{algorithm} [t!]
\fontsize{8}{6}\selectfont
\SetAlgoLined 
\SetKwInput{KwData}{Input}
\KwIn{$\mathrm{I}_{u}^{f_c}(t)$'s of all CVs in content delivery slot $t$} 
	Check if the requested contents are in the cache storage, if any requested content is not available, forward the request to upper layer for extraction from cloud \;
	Calculate SoI $\mathcal{T}_{\text{SoI}}^t$ using (\ref{SoI})\;
	Find eligible CV set $\mathcal{U}_{\text{val}}^t$ using Algorithm \ref{validCVSetExtraction} \;
	Find total number of VC to formulate, i.e., $W(t)$ using (\ref{VC_To_form}) \;
	Calculate eligible CVs', i.e., $u \in \mathcal{U}_{\text{val}}^t$, priorities using (\ref{deadline_weight_in_rt}) \;
	Get the CV set $\mathcal{U}_{\text{sch}}^t$ to schedule by picking the top-$W(t)$ $\phi_u(t)$s  \;
	Find optimal VC configuration $\mathcal{B}_{vc}^{a^{*}}(W(t))$ and optimal pRB allocations $\mathbf{I}_z^{{b,u}^{*}}(t)$ by running Algorithm \ref{getPRBAllocationAndVCConfig} \;
	Based on VC configuration $\mathcal{B}_{vc}^{a^{*}}(W(t))$ and $\mathbf{I}_z^{{b,u}^{*}}(t)$ calculate CVs SNRs $\pmb{\Gamma}_t = \{\Gamma_{u}^z(t) \}_{u \in \mathcal{U}_{\text{sch}}^t}$ using (\ref{V2I_SINR}) \;
	Calculate $R_u^{bit}(t)$ using (\ref{possileTxBitsPerState}) for all $u \in \mathcal{U}_{\text{sch}}^t$ \;
	Offload $R_u^{bit}(t)$ bits from the remaining payloads of all CVs $u \in \mathcal{U}_{\text{sch}}^t$ orderly from the requests made in the SoIs $\mathcal{T}_{\text{SoI}}^t$ \;
	Update all $u \in \mathcal{U}$ remaining payload and deadline \;
\caption{Content Delivery Model}
\label{deliveryModelAlgo}
\end{algorithm}

Upon receiving $\mathbf{R}_t$, the edge server leverages graph theory to get the optimal assignment as follows.
It forms a bipartite graph $G=(\mathcal{B} \times \mathcal{Z}, \mathcal{E})$, where $\mathcal{B}$ and $\mathcal{Z}$ are the set of vertices, and $\mathcal{E}$ is the set of edges that can connect the vertices \cite{west2001introduction}.
Moreover, $R_t[b,z]$ are the weights of edge $e(b,z)$ that connects vertex $b\in\mathcal{B}$ and $z\in \mathcal{Z}$.
Note that, for the graph $G$, a matching is a set of pair-wise non-adjacent edges where no two edges can share a common vertex. 
This is commonly known as the maximum weighted bipartite matching (MWBM) problem \cite{west2001introduction}.
The edge server needs to find the set of edges $e^{*}(b,z) \in \mathcal{E}$ that maximizes the summation of the weights of the edges. 
Moreover, the edge server uses well-known Hungarian algorithm \cite{kuhn1955hungarian} to get the optimal edges $e^{*}(b,z)$, i.e., pRB allocations $\mathrm{I}_{z}^{b,u}(t)$s in polynomial time.
This pRB allocation is, however, optimal only for the selected VC configuration $\mathcal{B}_{vc}^a(W(t))$.
In order to find the best VC configuration $\mathcal{B}_{vc}^{a^{*}} \left(W(t)\right)$, the edge server performs a simple linear search over all $A_{W(t)}$s VC configurations.
As such, we can solve problem (\ref{Sum_Rate_Maximization}) optimally using the above techniques.
Algorithm \ref{getPRBAllocationAndVCConfig} summarizes the steps.
Note that Algorithm \ref{getPRBAllocationAndVCConfig} has a time complexity of $\mathcal{O} \big(A_{W(t)}\big[W(t)(2Z|VC_a^i| + 1 ) + Z^3 + 4 \big] + 3 \big)$.

\vspace{-0.1in}
\subsection{Content Delivery Process}

\noindent
Contents are placed using the trained CPP $\pi_{ca}$ during each cache placement slot $t=n\Upsilon$, while the CVs make content requests in each $t$ following Section \ref{Content_Request_Modeling}.
Please note that, during $t=n\Upsilon$, the edge server only requires to perform one forward pass\footnote{The time complexity of the forward pass depends on the input/output size and DNN architecture.} on the trained $\pmb{\theta}_{ca}$.
Upon receiving the $\mathrm{I}_u^{f_c}(t)$s, the edge server checks whether $\mathrm{I}_{f_c}(n)=1$ or $\mathrm{I}_{f_c}(n)=0$.
If $\mathrm{I}_{f_c}(n)=1$, $f_c$ can be delivered locally.
All cache miss events are forwarded to the VEN's upper layers. 
The upper layer extracts each cache missed content from the cloud with an additional delay of $d_{u,f_c}^{m,t}$.
In all $t$, the edge server calculates the SoI $\mathcal{T}_{\text{SoI}}^t$ using (\ref{SoI}).
It then finds the eligible CV set $\mathcal{U}_{\text{val}}^t$ and forms total $W(t)$ VCs using Algorithm \ref{Extract_WSR_Matrix_per_given_Sch_CVs_VC_Selection} and (\ref{VC_To_form}), respectively. 
To that end, the edge server calculates the priorities $\phi_u(t)$s using (\ref{deadline_weight_in_rt}) and selects top-$W(t)$ CVs to schedule.
Once the edge server knows $W(t)$, $\phi_u(t)$s and $\mathcal{U}_{\text{sch}}^t$, it runs Algorithm \ref{getPRBAllocationAndVCConfig} to get the VC configuration and pRB allocations that maximizes the WSR of (\ref{Sum_Rate_Maximization}).
Algorithm \ref{getPRBAllocationAndVCConfig} returns the $\mathcal{B}_{vc}^{a^{*}}$ and $\mathbf{I}_{z}^{{b,u}^{*}}(t)$ which then can be used to get the SNRs $\Gamma_{u}^z(t)$s from (\ref{V2I_SINR}).
Upon receiving the SNRs $\Gamma_{u}^z(t)$s, the edge server can calculate the possible transmitted bits for the CVs as follows: 
\vspace{-0.1in}
\begin{equation}
\label{possileTxBitsPerState}
	R_u^{bit}(t) = \kappa \cdot R_u(t).
\end{equation}
The edge server then delivers the remaining $\mathrm{P}_{u,\text{rem}}^{t-d_f^{max}+\zeta}$s sequentially.
This entire process is summarized in Algorithm \ref{deliveryModelAlgo}.
The time complexity of running Algorithm \ref{deliveryModelAlgo} is $\mathcal{O} \big(U\big[\Lambda/S + 4|\mathcal{T}_{\mathrm{SoI}}^t| + 4\big] + W(t) \big[\log( |\mathcal{U}_{\mathrm{val}}^t|) + A_{W(t)} (2Z|VC_a^i| + 1 ) + 3 \big] + A_{W(t)}\big[Z^3 + 4 \big] + |\mathcal{U}_{\mathrm{val}}^t| + |\mathcal{T}_{\mathrm{SoI}}^t| + 10 \big)$.

\vspace{-0.1in}
\begin{table}[!t] 
	\caption{System Parameters }
	\centering
	\fontsize{7}{7}\selectfont
	\begin{tabular}{|c|c|}  \hline
		\textbf{Item/Description} & \textbf{Value} \\ \hline
		Total number of APs $B$ & $6$ \\ \hline
		Maximum possible VC per slot $\mathrm{W}_{\text{max}}$ & $5$ \\ \hline
		TTI $\kappa$ & $1$ ms \\ \hline
		DoI update interval $\Upsilon$ & $50 \times \kappa$ \\ \hline 
		Carrier frequency & $2$ GHz \\ \hline 
		pRB size $\omega$ & $180$ KHz \\ \hline
		Noise power $\sigma^2$ & -$174$ dBm/Hz \\ \hline
		AP coverage radius & $250$ m \\ \hline
		Antenna/AP $L$ & $4$ \\ \hline
		AP antenna height & $25$ m \\ \hline
		CV antenna height & $1.5$ m \\ \hline
		Transmission power $P_b$ & $30$ dBm \\ \hline 
		AP transmitter antenna gain $G_{TX}^b$ & $8$ dBi \\ \hline
		CV receiver antenna gain $G_{RX}^u$ & $3$ dBi \\ \hline
		CV receiver noise figure $L_{RX}^u$ & $9$ dB \\ \hline
		Total content class $c$ & $3$ \\ \hline 
		Contents per class $|\mathcal{F}_c|$ & $5$ \\ \hline
		Feature per content $G_c$ & $10$ \\ \hline
		AN cache size $\Lambda$ & $\{3,6,9,12\} \times S$  \\ \hline 
		Max allowable delay $d_f^{max}$ & $10 \times \kappa$ \\ \hline
		Content extraction delay $d_{u,f_c}^{m,t}$ & $5 \times \kappa$ \\ \hline
		CV active probability $p_u$ & Uniform$(0.1, 1)$ \\ \hline 
		CV's inclination to similarity/popularity $\epsilon_u$ & Uniform$(0, 1)$ \\ \hline
	\end{tabular}
	\label{System_Simulation_Params}
\end{table}
\begin{figure}[!t] \vspace{-0.15in}
    \begin{minipage}{0.24\textwidth}
    	\centering
    	\includegraphics[trim=110 0 0 0, clip, width=\textwidth]{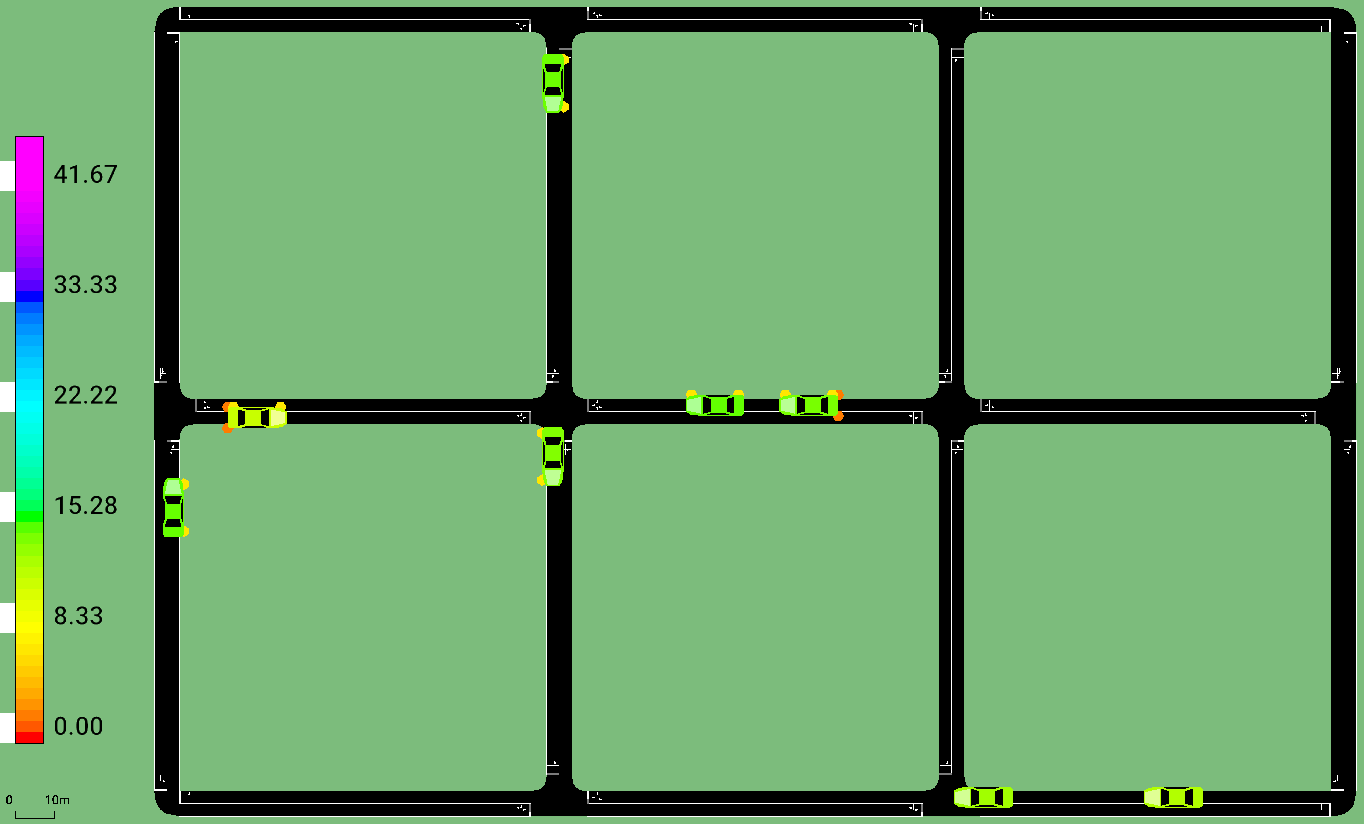}
    	\caption{Simulated RoI}
    	\label{simEnv}
    \end{minipage} 
    \begin{minipage}{0.24\textwidth}
    	\centering
    	\includegraphics[width=\textwidth]{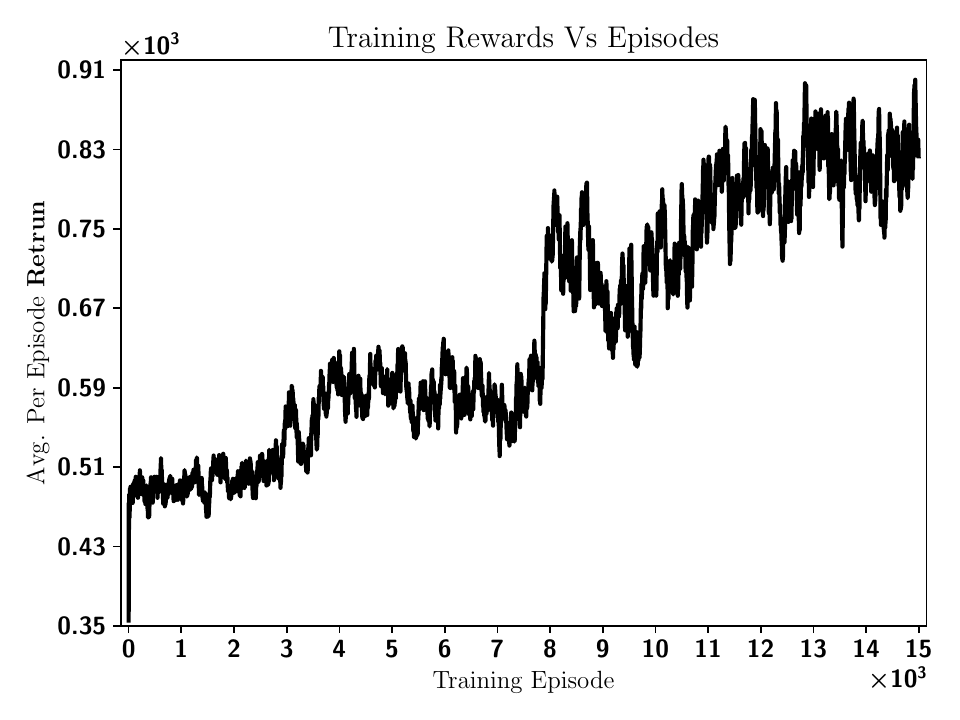}
    	\caption{CPP learning: average return during training}
    	\label{caTrainEpochVsReward}
    \end{minipage} 
\end{figure}

\section{Performance Evaluation}
\label{resultDiscussSection}
\subsection{Simulation Setting}
\noindent
We consider $U$ CVs roam over a region of interest (RoI) and deploy $B=6$ APs alongside the road to cover the entire RoI.
Table \ref{System_Simulation_Params} shows other key simulation parameters used in this paper.
We consider a $300$ meters by $200$ meters Manhattan grid model \cite{3GPP_TR_38_886} with two-way roads as shown in Fig. \ref{simEnv}.
For realistic microscopic CV mobility modeling, we use the widely known simulation of urban mobility (SUMO) \cite{8569938}.
The CVs are deployed with some initial routes with a maximum speed of $45$ miles/hour and later randomly rerouted from the intersections on this RoI. 
In SUMO, we have used car-following mobility model \cite{roy2011handbook} and extracted the CVs' locations using the Traffic Control Interface \cite{wegener2008traci} application programming interface.

To design our simulation episode, we consider $1000 \kappa$ milliseconds of CVs activities. 
For the CPP learning, the edge server uses DNN $\pmb{\theta}_{ca}$ that has the following architecture: $2$D convolution (\textit{Conv$2$d)} $\rightarrow$ \textit{Conv$2$d} $\rightarrow$ \textit{Linear} $\rightarrow$ \textit{Linear}.
We train $\pmb{\theta}_{ca}$ in each cache placement slots with a batch size $S_{ca}=512$.
Besides, we choose $\gamma=0.995$, $\epsilon_{max}=1$, $\epsilon_{min}=0.005$, $\nu=0.6$, $Mem_{ca}^{max}=15000$, $T_{epoch}=15000$, $\breve{\eta}_{ca}=4\Upsilon$.
For training, we use $Adam$ as the optimizer with a learning rate of $0.001$. 
Using our simulation setup, the edge server first learns $\pi_{ca}$ using Algorithm \ref{CPP_Algorithm} for $T_{epoch}$ episodes.
The average per state returns during this learning is shown in Fig. \ref{caTrainEpochVsReward}. 
As the training progresses, we observe that the edge server learns to tune its policy to maximize the expected return. 
After sufficient exploration, the edge server is expected to learn the CPP that gives the maximized expected return.
As a result, it is expected that the reward will increase as the learning proceeds. 
Fig. \ref{caTrainEpochVsReward} also validates this and shows the convergence of Algorithm \ref{CPP_Algorithm}.
As such, we use this trained CPP $\pi_{ca}$ for performance evaluations in what follows.

\vspace{-0.1in}
\subsection{Performance Study}
\noindent
We first show the performance comparisons of the learned CPP with the following baselines without any RAT solution.

\textbf{Genie-Aided cache replacement} (Genie): The to-be requested contents are known beforehand during the start of the DoI provided by a Genie. In this best case, we then store the top-$\Lambda^c$ requested contents from all $c\in \mathcal{C}$ in all $n$.

\textbf{Random cache replacement (RCR)}: In this case, contents from each class are selected randomly for cache placement.

\textbf{$K$-Popular ($K$-PoP) replacement} \cite{8809280}: In this popularity-based caching mechanism, we store the most popular $K=\Lambda^c$ contents during the past DoI for each content class  $c\in\mathcal{C}$. 

\textbf{Modified $K$-PoP+LRU ($K$-LRU) replacement}: We modify the popularity-driven $K$-PoP with classical least recently used (LRU) \cite{cacheRepSurvey} cache replacement. The least popular contents in the $K$-PoP contents are replaced by the most recently used but not in $K$-PoP contents to prioritize recently used contents.

To this end, we vary the cache size of the VEN and show the average CHR during an episode in Fig. \ref{cacheSizeVsCHRWoRAT}. 
The general intuition is that when we increase the cache size $\Lambda$, more contents can be placed locally. 
Therefore, by increasing $\Lambda$, the average CHR is expected to increase. 
$K$-PoP and $K$-LRU do not capture the heterogeneous preferences of the CVs.
Similarly, as contents are replaced randomly with the naive RCR baseline, it should perform poorly.
However, when the cache size is relatively small, solely popularity-based $K$-PoP performs even worse than RCR.
This means that popularity does not dominate the content demands of the CVs.
Moreover, when the cache size becomes moderate, $K$-PoP and $K$-LRU outperform the naive RCR baseline.
On the other hand, the proposed CPP aims to optimize $\pi_{ca}$ by capturing the underlying preference-popularity tradeoff of the CVs.
Therefore, the average CHR is expected to be better than the baselines.
Fig. \ref{cacheSizeVsCHRWoRAT} also reflects these analysis. 
Moreover, notice that the performance gap with the Genie-aided average CHR and our proposed CPP is lower. 
In the VEN, we do not know the future and CVs' content demands. 
Therefore, we can only predict the future and tune the CPP $\pi_{ca}$ accordingly. 
Particularly, when the cache storage is reasonable, the performance gap of the proposed CPP is much lower. For example, at $\Lambda=9$ and $\Lambda=12$ the proposed CPP delivers around $93\%$ and $98\%$ of the Genie-aided solution.
Moreover, the baselines perform poorly regardless of $\Lambda$. For example, at $\Lambda=9$, the proposed CPP is around $49\%$, $23 \%$ and $24\%$ better than RCR, $K$-PoP and $K$-LRU, respectively.

\begin{figure*} \vspace{-0.1in}
\begin{minipage}{0.33\textwidth}
	\centering
	\includegraphics[trim=12 5 40 10, clip, width=\textwidth]{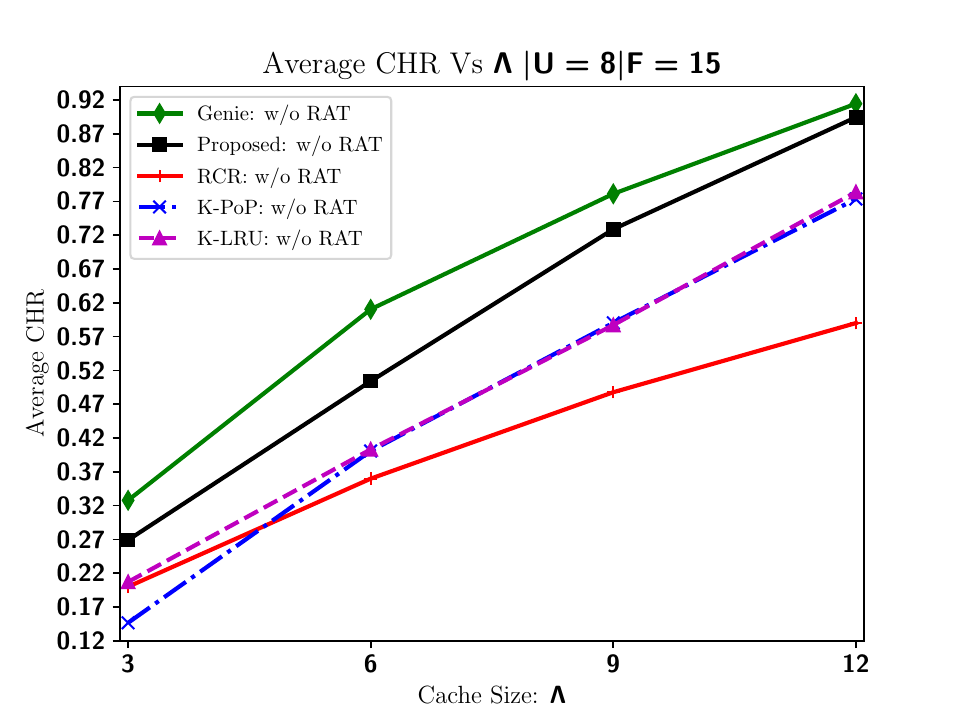} \vspace{-0.2in}
	\caption{CHR comparison with baselines when $\Upsilon=50 \times \kappa$ (without RAT)}
	\label{cacheSizeVsCHRWoRAT}
\end{minipage}  \hspace{0.01in}
\begin{minipage}{0.33\textwidth}
	\centering
	\includegraphics[trim=12 5 40 10, clip, width=\textwidth]{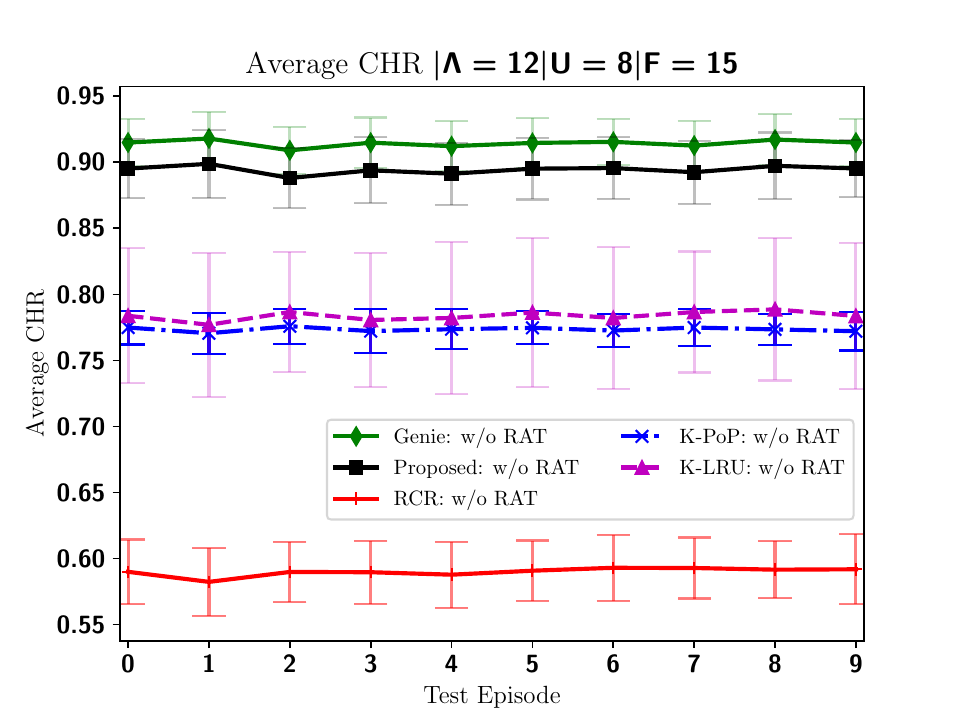} \vspace{-0.2in}
	\caption{CHR comparison with baselines for $10$ test episodes when $\Upsilon=50\times \kappa$ (without RAT)}
	\label{chrVsEpochWoRAT}
\end{minipage} \hspace{0.01in}
\begin{minipage}{0.33\textwidth}
	\centering
	\includegraphics[trim=12 5 38 10, clip, width=\textwidth]{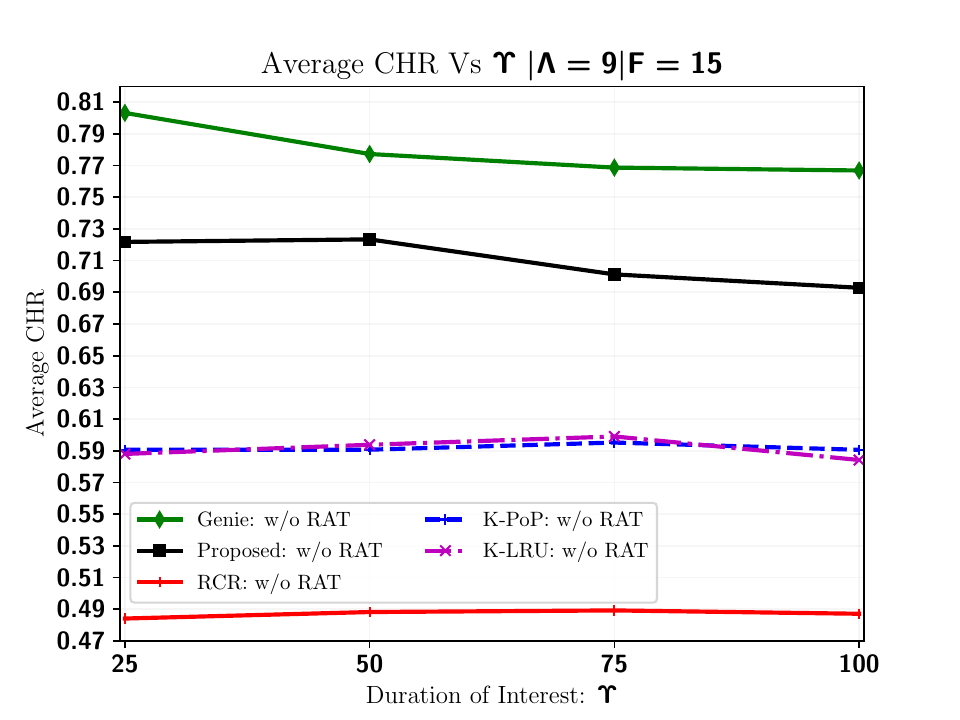} \vspace{-0.2in}
	\caption{CHR analysis for different DoI when $U=8$ (without RAT)}
	\label{DoIVsCHRWoRAT}
\end{minipage}
\end{figure*}
\begin{figure} \vspace{-0.2in}
	\centering
	\includegraphics[trim=10 5 40 10, clip, width=0.4\textwidth]{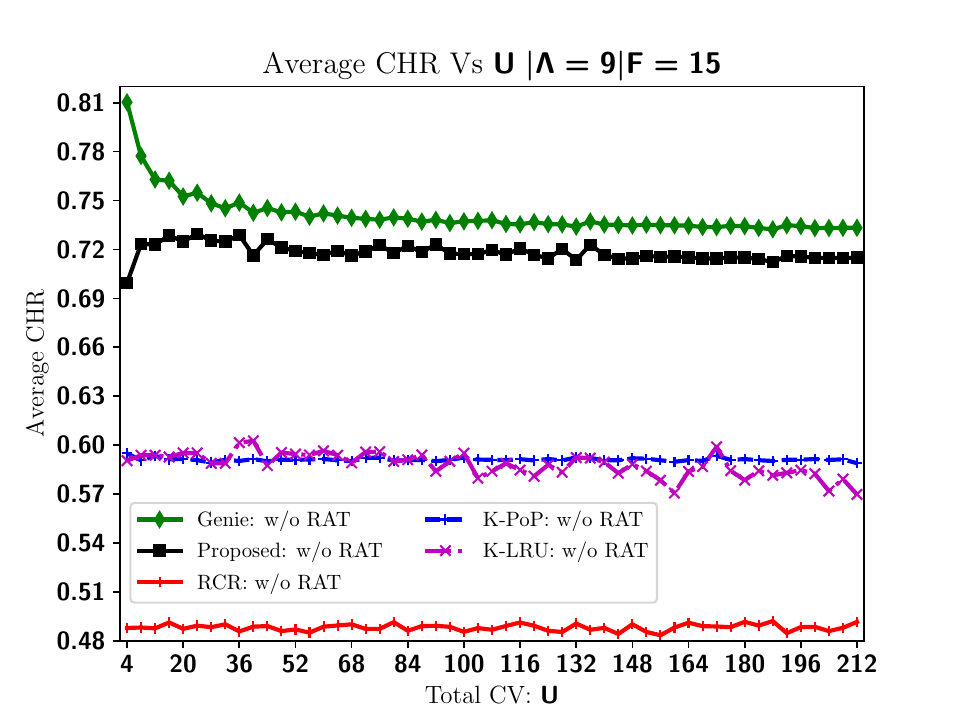} \vspace{-0.08in}
	\caption{CHR comparison with baselines for different CVs when $\Lambda=9$ (without RAT)}
	\label{totalCVvsCHRWoRAT}
\end{figure}

Fig. \ref{chrVsEpochWoRAT} shows the average CHR variation over $10$ test episode in $100$ simulation runs and corresponding standard deviations.
As expected, the performance of the proposed CPP is very close to the Genie-aided performance in these test runs.
Particularly, the proposed CPP delivers around $98\%$ of the Genie-aided performance.
Moreover, the other baselines' average CHRs largely deviate from the Genie-aided solution.
We observe that the proposed CPP is around $52\%$, $16\%$ and $14\%$ better than RCR, $K$-PoP and $K$-LRU, respectively, even when $\Lambda$ is $80\%$ of the content catalog $\mathcal{F}$, which validate the effectiveness of the proposed method.

To this end, we study the impact of different DoI $\Upsilon$ on the CHR. 
Recall that the DoI is the period for which the contents in the library remain fixed. 
A shorter DoI means that the content catalog can be refreshed quickly. 
Besides, based on our content request model, each CV's content choices change fewer times within this short interval.
Hence, the edge server can quickly accommodate the CPP to capture the future demands of the CVs.
This, thus, may yield better CHR.
On the other hand, when this period is extended, performance is expected to deteriorate slightly.
This is due to the fact that the cache storage cannot be replaced until this DoI period expires, while the CVs' requests vary in each slot.
We also observe similar trends in our simulation results.
Fig. \ref{DoIVsCHRWoRAT} shows CHR for different DoI, where we observe that even the Genie-aided performance degrades from $80\%$ to $76\%$ when the DoI is increased from $25\times \kappa$ to $100 \times \kappa$.
We also observe that our proposed CPP experiences only about $4\%$ performance degradation.
Moreover, the performance improvement of our proposed solution is about $49\%$, $22\%$ and $23\%$ at $\Upsilon=25 \times \kappa$ and about $42\%$, $17\%$ and $18\%$ at $\Upsilon=100 \times \kappa$, respectively, over the RCR, $K$-Pop and $K$-LRU baselines.
Note that we leave the choice of DoI as a design parameter chosen by the system administrator, which can be decided based on the practical hardware limitations and other associated overheads in the network. As such, we fix $\Upsilon=50\times \kappa$ for the rest of our analysis.

As content requests arrive following preference-popularity tradeoff, the CHR also gets affected by the total number of CVs in the VEN. 
Intuitively, as the CVs' preferences are heterogeneous, when the total number of CVs in the VEN increases, the content requests largely diversify. 
Therefore, even with the Genie-aided solution, the CHR may degrade when the number of CVs in the VEN increases. 
This is also reflected in our simulated results in Fig. \ref{totalCVvsCHRWoRAT}.
The performance of the proposed CPP algorithm is stable regardless of the number of CVs in the VEN. 
We observe a slight performance gap between the CPP and the Genie-aided solution.
This gap gets smaller and smaller as the total number of CVs in the VEN increases.
Particularly, we observe that the proposed CPP delivers an average $97 \%$ CHR for the considered CV numbers.
Besides, it delivers around $47 \%$, $21 \%$ and $22\%$ better performance than RCR, $K$-PoP and $K$-LRU, respectively.
Therefore, we will use this CPP $\pi_{ca}$ to find $\mathrm{I}_{f_c}(n)$ for all $n$ and show performance analysis of our proposed user-centric RAT solution.

To that end, we compare the performance of the proposed RAT solution with legacy network-centric RAT (NC-RAT).
In the NC-RAT, a base station (BS) is located at a fixed suitable location which has $Z=6$ pRBs and total transmission power of $46$ dBm. 
We use the same scheduling and deadline-based priority modeling for the NC-RAT as the proposed user-centric case.
Besides, we distributed the total transmission power proportionally to the scheduled CVs' priorities.
Moreover, we performed the same WSR maximization problem for getting the pRB allocation using Hungarian algorithm \cite{kuhn1955hungarian}.
In the following, this legacy RAT solution is termed NC-RAT and used with the cache placement baselines.
On the other hand, the `Proposed' method uses the proposed CPP and user-centric RAT solution.

Intuitively, with an increased $\Lambda$, the edge server can store more contents locally which increases the total number of local delivery by assuring lower cache miss events.
Therefore, with a proper RAT solution, the content delivery delay is expected to decrease if we increase the cache size of the edge server. 
We also observe this trend with both NC-RAT and our proposed user-centric RAT solution in Fig. \ref{delayCompRATs}.
However, note that NC-RAT is inflexible, and depending on the location of the CVs, NC-RAT may not even have expected radio-link qualities. 
This can, therefore, cause link failure and may increase the content delivery delay for the CVs' requested content. 
On the other hand, the proposed user-centric RAT solution can design the appropriate VC configuration, VC associations and proper radio resource allocation to deliver the content timely. 
Therefore, we expect the user-centric RAT solution to outperform the traditional NC-RAT.
Fig. \ref{delayCompRATs} shows the average content delivery delay $\mathrm{d}=\frac{1}{T}\sum_{t=1}^T \mathrm{\bar{d}}(t)$, where $\mathrm{\bar{d}}(t)$ is calculated in (\ref{Average_Delay}) with $\mathrm{W}_{\text{max}}=5$.
As we can see, the proposed solution outperforms the baselines.
Particularly, the average gain of the proposed solution on content delivery delay is around $15\%$ over the baselines.

The effectiveness of the proposed solution is more evident in Fig. \ref{deadlineViolationsRATs}, which shows the percentage of deadline violations in a test episode when the content size is $S=4$ KB.
As a general trend, the deadline violations decrease as $\Lambda$ increases. 
Besides, among the cache placement baselines, as we have seen in the performance comparisons of the CPP, even RCR delivers lower deadline violations than solely popularity-based $K$-PoP when the cache size is small.
Moreover, we observe around $28 \%$ higher deadline violations with the baseline NC-RAT over our proposed user-centric RAT solution.
Recall that this deadline violation is essentially the violation of constraint $C_{12}$, which means the requester CVs have not received the requested content by their required deadlines.
As such, these requester CVs may experience fatalities and degraded QoEs with the existing RAT and cache placement baselines.

Content size $S$ also affects the delivery delays and corresponding deadline violations.
Intuitively, content delivery delay shall increase if the payload increases when the network resources are unchanged. 
This also increases the likelihood of deadline violations.
Fig. \ref{dealyContentSizeRAT} shows how the delivery delay gets affected by content size $S$. 
Note that transmission delay is directly related to channel quality between the transmitter and receiver.
This channel uncertainty can cause fluctuations in the content delivery delay.
However, the general expectation is that the content delivery delay will increase when the payload size increases.
We also observe these in Fig. \ref{dealyContentSizeRAT}.
Particularly, when $S=2.5$ KB, the performance gain of the proposed solution is  around $30 \%$ over the RCR+NCRAT and around $27 \%$ over the $K$-PoP+NCRAT and $K$-LRU+NCRAT baselines.

Recall that delay cannot exceed the hard deadline. 
Therefore, higher content delivery delay leads to deadline violations. 
Fig. \ref{violationVsContentSize} shows how the payload size affects the deadline violations in the proposed VEN.
As expected, even when the payload size is small, we observe that the legacy NC-RAT solution cannot ensure guaranteed delivery within the deadline.
On the contrary, our proposed solution can ensure $0 \%$ deadline violations till $S=3$ KB. 
Moreover, when $S$ increases, the deadline violation percentage of our proposed solution performs significantly better than the NCRAT-based baselines.
For example, when $S=4$ KB, the deadline violation percentage with our proposed solution is around $12 \%$, whereas the NCRAT-based baselines have around $47\%$ deadline violations.
From Fig. \ref{delayCompRATs} - Fig. \ref{violationVsContentSize}, we can clearly see that the traditional NC-RAT is not sufficient to deliver the demands of the CVs.

To that end, we show the efficacy of the proposed RAT solution by considering all cache placement baselines accompanied by the proposed RAT solution for delivering the requested contents of the CVs.
Fig. \ref{delayUCRATCacheSize} shows how the content delivery delay gets affected by different cache sizes. 
Particularly, the proposed solution delivers requested contents around $14\%$, $7\%$ and $8\%$ faster than the RCR+Proposed-RAT, $K$-PoP+Proposed-RAT and $K$-LRU+Proposed-RAT, respectively, when $\Lambda=9$.
Recall that the proposed CPP (without RAT) had a performance gain of around $49\%$, $23\%$ and $24\%$ over the RCR, $K$-LRU and $K$-PoP, respectively. 
The proposed RAT can, thus, significantly compensate for the cache miss events.

Moreover, Fig. \ref{delayUCRATTotalCVs} shows delay vs total number of CVs $U$ in the VEN.
Intuitively, if $U$ increases, the edge server receives a larger number of content requests. 
Then, with the limited VCs, the edge server can at max schedule only $\mathrm{W}_{\text{max}}$ number of CVs. 
Therefore, $\mathrm{d}$ is expected to increase if $U$ increases, which is also reflected in Fig. \ref{delayUCRATTotalCVs}. 
Notice that in both Fig. \ref{delayUCRATCacheSize} and Fig. \ref{delayUCRATTotalCVs}, while the proposed solution outperforms the other cache placement baselines, the performance gaps are small because all cache placement baselines now use our proposed user-centric RAT solution for delivering the requested contents.

\begin{figure*} \vspace{-0.1in}
\begin{minipage}{0.33\textwidth}
	\centering
	\includegraphics[trim=12 5 40 10, clip, width=\textwidth]{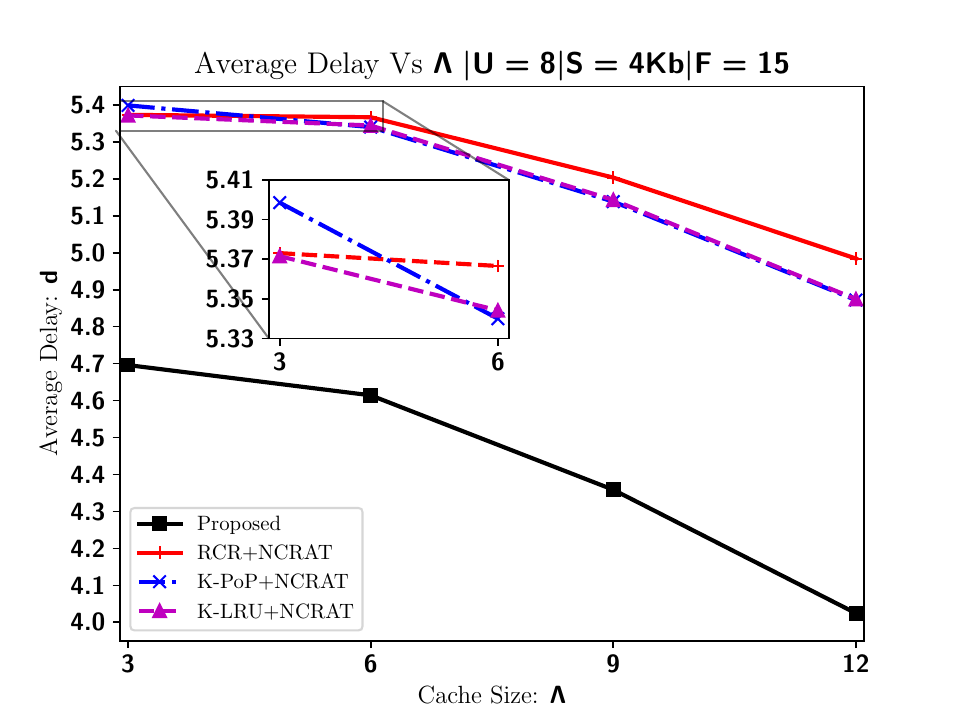} \vspace{-0.2in}
	\caption{Average content delivery delay comparison with NC-RAT and caching baselines}
	\label{delayCompRATs}
\end{minipage} \hspace{0.01in}
\begin{minipage}{0.33\textwidth}
	\centering
	\includegraphics[trim=12 5 40 10, clip, width=\textwidth]{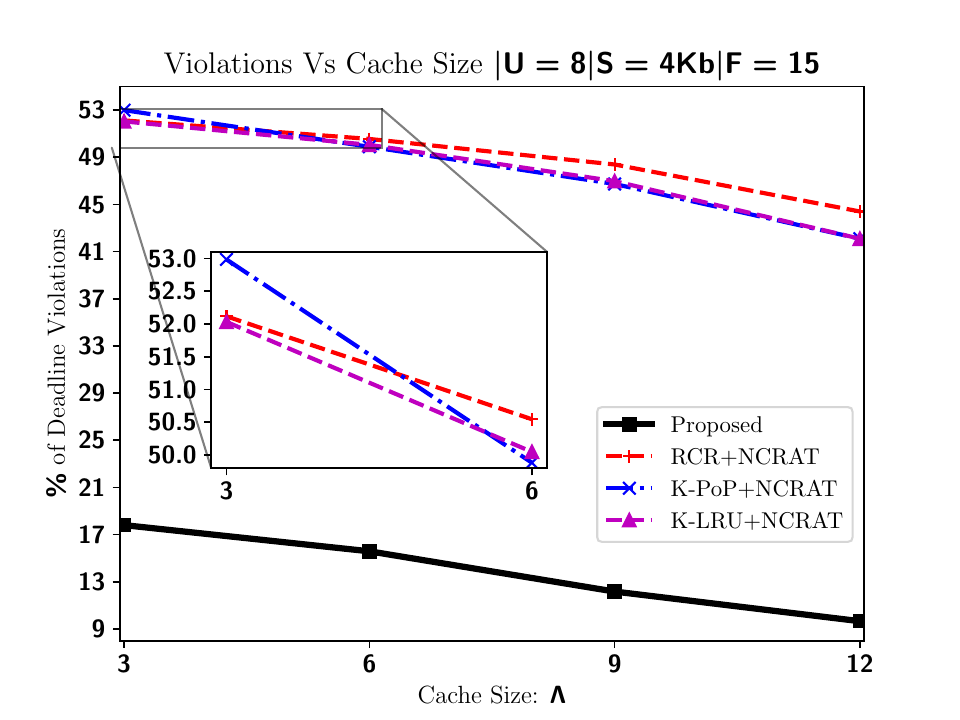} \vspace{-0.2in}
	\caption{Deadline violation percentage comparison with legacy NC-RAT}
	\label{deadlineViolationsRATs}
\end{minipage}  \hspace{0.01in}
\begin{minipage}{0.33\textwidth}
	\centering
	\includegraphics[trim=12 5 35 10, clip, width=\textwidth]{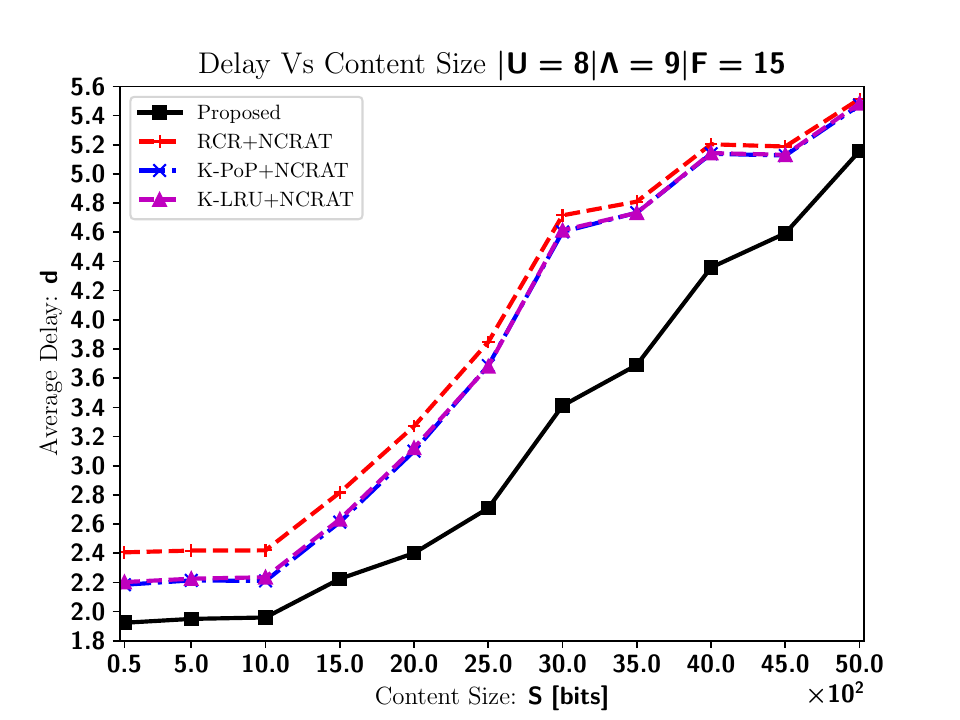} \vspace{-0.2in}
	\caption{Average content delivery delay comparisons for different content sze $S$}
	\label{dealyContentSizeRAT}
\end{minipage}
\begin{minipage}{0.33\textwidth}
	\centering
	\includegraphics[trim=12 5 35 10, clip, width=\textwidth]{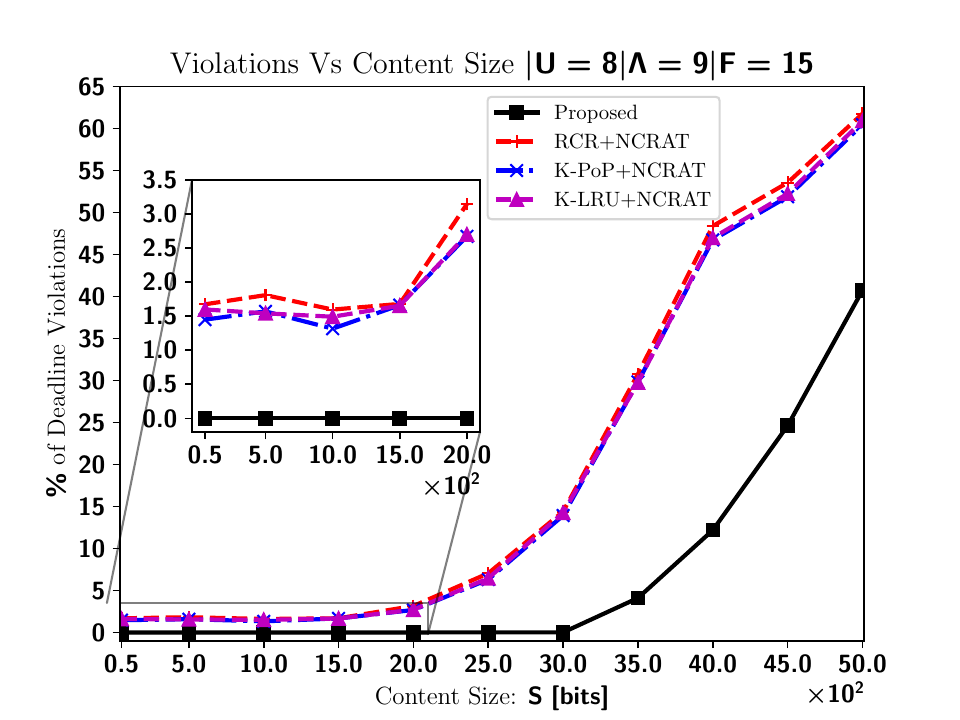} \vspace{-0.2in}
	\caption{Deadline violation percentage comparisons for different content size $S$} 
	\label{violationVsContentSize}
\end{minipage} \hspace{0.01in}
\begin{minipage}{0.33\textwidth}
	\centering
	\includegraphics[trim=12 8 40 10, clip, width=\textwidth]{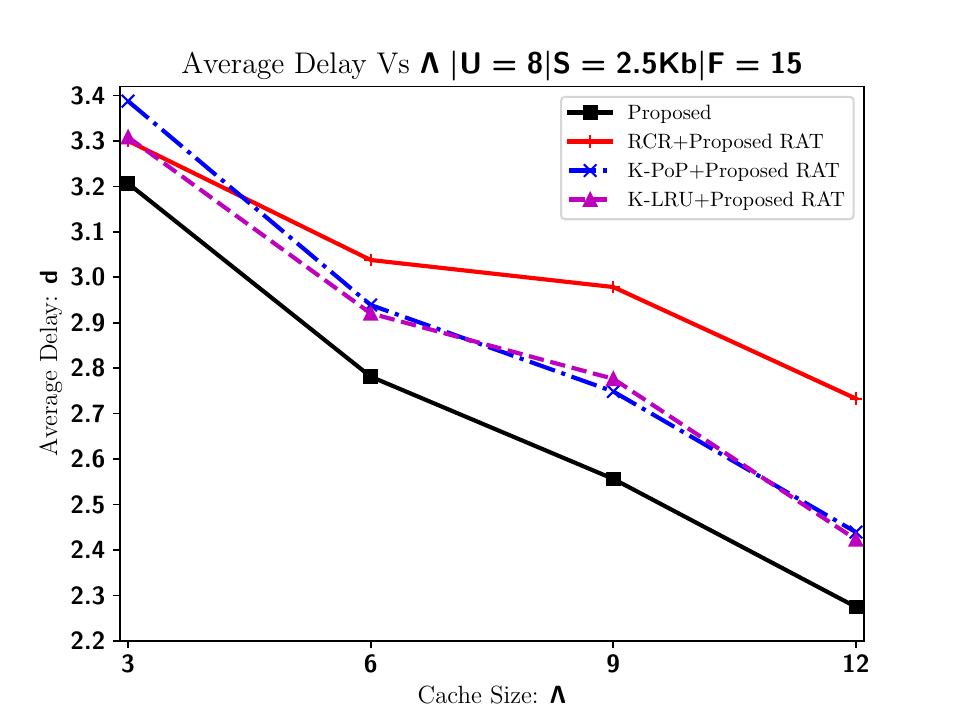} \vspace{-0.2in}
	\caption{Average content delivery delay for different $\Lambda$} 
	\label{delayUCRATCacheSize}
\end{minipage} \hspace{0.01 in}
\begin{minipage}{0.33\textwidth}
	\centering
	\includegraphics[trim=12 5 40 10, clip, width=\textwidth]{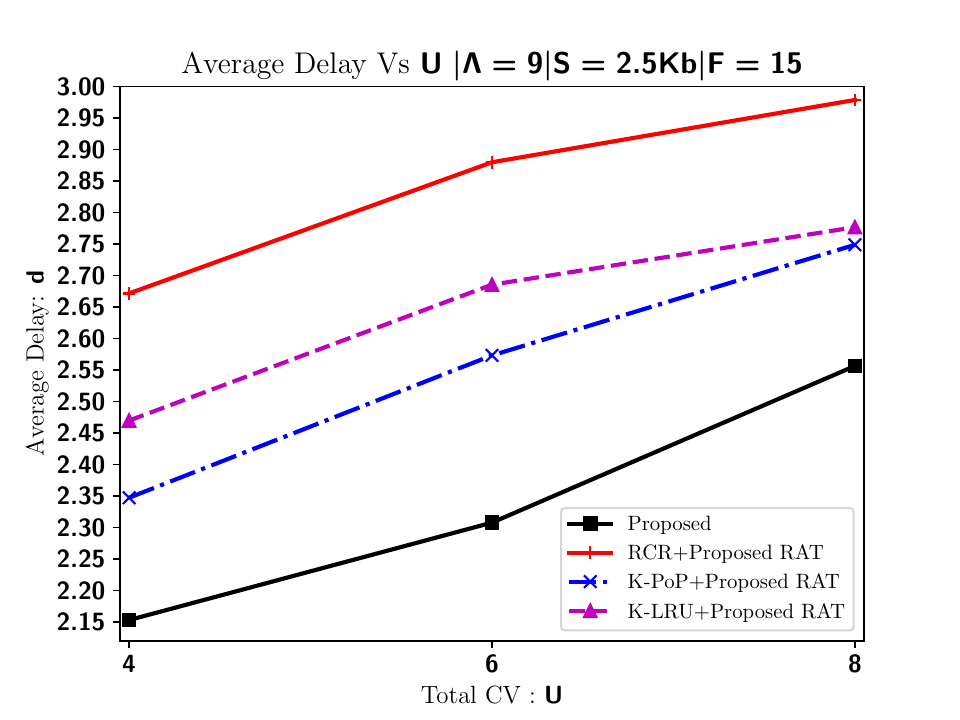} \vspace{-0.2in}
	\caption{Average delay for different number of CVs}
	\label{delayUCRATTotalCVs} 
\end{minipage} 
\end{figure*}

Finally, our extensive simulation results suggest that the CHR from our proposed CPP is very close to the genie-aided solution, while the baseline RCR, $K$-Pop and $K$-LRU cache placements yield poor CHRs. 
Besides, when we use the proposed CPP with our user-centric RAT solution, the performance improvements, in terms of deadline violation percentage and content delivery delay, are significant compared to existing legacy NC-RAT with the above cache placement baseline solutions.
Additionally, having a larger cache storage size increases the CHR, while having more CVs in the VEN leads to a slightly degraded CHR for all cache placement strategies.
Moreover, with fixed limited radio resources, content delivery delays grow, which increases the deadline violation percentage.

\vspace{-0.1in}
\section{Conclusion}
\label{conclusion}
\noindent
Considering the higher automation demand on the road, in this paper, we propose a user-centric RAT solution for delivering the CVs requested content with a learning solution for the cache placement.
From the results and analysis, we can conclude that existing cache placement baselines may not be sufficient to capture the heterogeneous demands and preferences of the CVs.
Moreover, the existing NC-RAT may cause severe fatalities on the road as it yields frequent deadline violations. 
Even for continuous deadline-constrained demand arrivals in each time slot, the proposed software-defined user-centric RAT solution has shown significant potential for offloading the payloads timely.
The results suggest that our proposed cache placement policy delivers practical near-optimal cache hit ratio while the proposed user-centric RAT efficiently delivers the requested contents within the allowable deadline.

\vspace{-0.1in}

\appendix
\vspace{-0.05in}
\subsection{Proof of Proposition  \ref{Prop_Chernoff_Bound}}
\label{Appendix_A}
\noindent
Assuming $\iota>0$, we write the following:
\begin{subequations}
\label{Proof_Chernoff_Bound}
\begin{align}
	& \mathrm{Pr}\left\{\Psi_t \geq \xi \right\}  = \mathrm{Pr} \big\{e^{\iota\Psi_t} \geq e^{\iota\xi}\big\} 
	\stackrel{(a)}{\leq} (\mathbb{E}\big[e^{\iota\Psi_t} \big])/e^{\iota\xi}, \\
	& \stackrel{(b)}{=} e^{-\iota\xi} \prod\nolimits_{u=1}^U \mathbb{E}\left[e^{\iota \Theta_u^t} \right]  
	\stackrel{(c)}{=} e^{-\iota\xi} \prod\nolimits_{u=1}^U \left(1 - p_u + p_u e^\iota\right), \\
	& \stackrel{(d)}{\leq} e^{-\iota\xi} \left[ \frac{\sum_{u=1}^U \left(1 - p_u + p_u e^\iota\right)}{U}\right]^U \eqz \eqz \eqz  
	= e^{-\iota\xi} \left[1 - \bar{p} + \bar{p} e^{\iota} \right]^U\eqz,\eqz \\
	&= \exp\left[-\iota \xi  + U \ln\left(1-\bar{p}+\bar{p}e^\iota\right) \right], 
\end{align}	
\end{subequations}
where $\left(a\right)$ follows Markov inequality, $\left(b\right)$ is true as $\Theta_u^t$s are independent and identically distributed, $\left(c\right)$ follows as $\mathbb{E}\big[e^{\iota \Theta_u^t}\big]$ is the moment generating function of $\Theta_u^t$, and $\left(d\right)$ is obtained following the inequality of arithmetic and geometric means.

To this end, we find $e^\iota = \frac{\xi \left(1-\bar{p}\right)}{\bar{p}\left(U-\xi\right)}$ that minimizes (\ref{Proof_Chernoff_Bound}).
Plugging this value in (\ref{Proof_Chernoff_Bound}), we obtain the bound as 
\begin{subequations}
\label{eq40}
\begin{align}
	&\mathrm{Pr}\left\{\Psi_t \geq \xi \right\} 
    \stackrel{(a)}{\leq}
    \exp \eqz \left[U \left\{\eqz \ln \eqz \left(\eqz \frac{\left(1-\bar{p}\right)}{1-\chi}
	\eqz\right)\eqz - \chi \ln \left(\frac{\chi \left(1-\bar{p}\right)}{\bar{p}\left(1-\chi\right)} \right) \eqz \right\} \! \eqz \right], \nonumber\\
    & \quad= \exp\left[-U D_{\bar{p}} \left(\chi\right)\right], \tag{\ref{eq40}}
\end{align}
\end{subequations}
where $\chi = \frac{\xi}{U}$ in $(a)$ and $D_{\bar{p}} = \chi \ln \left(\frac{\chi}{\bar{p}} \right) + \left(1 - \chi \right) \ln \left( \frac{1-\chi}{1-\bar{p}}\right)$.

\vspace{-0.1in}
\subsection{Proof of Theorem  \ref{Thm_CHR_NP_Hardness}}
\label{Appendix_CHR_NP_Hardness}
\noindent
We show that an instance of our problem in (\ref{Cache_Hit_Max_Problem}) reduces to an instance of a well-known NP-hard problem.
Particularly, we only consider a single cache placement step $t=n\Upsilon$ and assume that $\mathrm{I}_u^{f_c}(t)$s, $\forall t \in [n\Upsilon, (n+1)\Upsilon]$ are known at the edge server beforehand\footnote{This assumption is only for the sake of this proof. The edge server does not know the future.}.
Then, we re-write our (\ref{Cache_Hit_Max_Problem}) instance as 
\begin{subequations}
\label{CHR_Instance}
\begin{align}
	\tag{\ref{CHR_Instance}} & \underset{ \mathrm{I}_{f_c}(n) ;~ \forall f_c \in \mathcal{F}} {\text{maximize}} ~~ \sum\nolimits_{t \in [n\Upsilon, (n+1)\Upsilon]}  \mathrm{CHR}(t), \\
	\label{CHR_Instance_Cons_1} & \sum\nolimits_{c=1}^C \sum\nolimits_{f_c \in \mathcal{F}_c} S \cdot  \mathrm{I}_{f_c}(n) \leq \Lambda, \sum\nolimits_{f_c \in \mathcal{F}_c} S \cdot \mathrm{I}_{f_c}(n) = \Lambda^c, \\
	\label{CHR_Instance_Cons_3} &  \mathrm{I}_{f_c}(n) \in \{0,1\}, \forall c=1,\dots, C; f_c \in \mathcal{F}_c,
\end{align}
\end{subequations}
where the constraints are taken for the same reasons as in (\ref{Cache_Hit_Max_Problem}).

To that end, if $\Lambda^c = S \cdot 1$, we could rewrite the second constraint as $\sum_{f_c \in \mathcal{F}_c} \mathrm{I}_{f_c}(n) = 1$.
Then, it is easy to recognize that an instance of the well-known multiple-choice knapsack problem (MCKP) \cite{kellerer2004multidimensional} has reduced to this instance of our CHR maximization problem.
As MCKP is a well-known NP-hard problem \cite{kellerer2004multidimensional}, we conclude that the cache placement problem for each $t=n\Upsilon$ is NP-hard even when the to-be requested contents are known beforehand. 
As such, the long-term policy optimization problem in (\ref{Cache_Hit_Max_Problem}) is NP-hard.

\vspace{-0.1in}
\bibliography{Reference.bib}

\begin{thebibliography}{10}
\providecommand{\url}[1]{#1}
\csname url@samestyle\endcsname
\providecommand{\newblock}{\relax}
\providecommand{\bibinfo}[2]{#2}
\providecommand{\BIBentrySTDinterwordspacing}{\spaceskip=0pt\relax}
\providecommand{\BIBentryALTinterwordstretchfactor}{4}
\providecommand{\BIBentryALTinterwordspacing}{\spaceskip=\fontdimen2\font plus
\BIBentryALTinterwordstretchfactor\fontdimen3\font minus
  \fontdimen4\font\relax}
\providecommand{\BIBforeignlanguage}[2]{{%
\expandafter\ifx\csname l@#1\endcsname\relax
\typeout{** WARNING: IEEEtran.bst: No hyphenation pattern has been}%
\typeout{** loaded for the language `#1'. Using the pattern for}%
\typeout{** the default language instead.}%
\else
\language=\csname l@#1\endcsname
\fi
#2}}
\providecommand{\BIBdecl}{\relax}
\BIBdecl

\bibitem{NHTSA_1}
\BIBentryALTinterwordspacing
``{Automated Vehicles for Safety},'' united States Department of
  Transportation, National Highway Traffic Safety Administration, {Accessed:}
  \today. [Online]. Available:
  \url{https://www.nhtsa.gov/technology-innovation/automated-vehicles-safety}
\BIBentrySTDinterwordspacing

\bibitem{DfT_UK}
\BIBentryALTinterwordspacing
``{Connected and Automated Vehicles in the UK: 2020 information booklet},''
  centre for Connected and Autonomous Vehicles, Department for Transportation,
  UK, {Accessed:} \today. [Online]. Available:
  \url{https://www.gov.uk/government/publications/connected-and-automated-vehicles-in-the-uk-2020-information-booklet}
\BIBentrySTDinterwordspacing

\bibitem{li2021federated}
X.~Li, L.~Cheng, C.~Sun, K.-Y. Lam, X.~Wang, and F.~Li,
  ``Federated-learning-empowered collaborative data sharing for vehicular edge
  networks,'' \emph{IEEE Network}, vol.~35, no.~3, pp. 116--124, 2021.

\bibitem{liu20206g}
M.~Noor-A-Rahim \emph{et~al.}, ``6g for vehicle-to-everything (v2x)
  communications: Enabling technologies, challenges, and opportunities,''
  \emph{Proc. IEEE}, vol. 110, no.~6, pp. 712--734, June 2022.

\bibitem{8734737}
Y.~F. {Payalan} and M.~A. {Guvensan}, ``Towards next-generation vehicles
  featuring the vehicle intelligence,'' \emph{IEEE Trans. Intelligent
  Transport. Syst.}, vol.~21, no.~1, pp. 30--47, June 2020.

\bibitem{8515151}
J.~Liu and J.~Liu, ``Intelligent and connected vehicles: Current situation,
  future directions, and challenges,'' \emph{IEEE Commun. Stand. Mag.}, vol.~2,
  no.~3, pp. 59--65, Sept. 2018.

\bibitem{chandupatla2020augmented}
S.~Chandupatla, N.~Yerram, and B.~Achuthan, ``Augmented reality projection for
  driver assistance in autonomous vehicles,'' SAE Technical Paper, Tech. Rep.,
  2020.

\bibitem{Pervej_UPLEdge}
{M. F. {Pervej}}, {L. T. {Tan}}, and {R. Q. {Hu}}, ``User preference learning
  aided collaborative edge caching for small cell networks,'' in \emph{Proc.
  IEEE Globecom}, Dec. 2020.

\bibitem{8367785}
I.~{Parvez}, A.~{Rahmati}, I.~{Guvenc}, A.~I. {Sarwat}, and H.~{Dai}, ``{A
  Survey on Low Latency Towards 5G: RAN, Core Network and Caching Solutions},''
  \emph{IEEE Commun. Surveys Tutor.}, vol.~20, no.~4, pp. 3098--3130, May 2018.

\bibitem{Pervej_AIACEC}
M.~F. Pervej, L.~T. Tan, and R.~Q. Hu, ``Artificial intelligence assisted
  collaborative edge caching in small cell networks,'' in \emph{Proc. IEEE
  Globecom}, Dec. 2020.

\bibitem{9552606}
L.~Zhao, H.~Li, N.~Lin, M.~Lin, C.~Fan, and J.~Shi, ``Intelligent content
  caching strategy in autonomous driving toward 6g,'' \emph{IEEE Trans. Intell.
  Transport. Syst.}, pp. 1--11, Sept. 2021.

\bibitem{9382020}
W.~Qi, Q.~Li, Q.~Song, L.~Guo, and A.~Jamalipour, ``Extensive edge intelligence
  for future vehicular networks in 6g,'' \emph{IEEE Wireless Commun.}, vol.~28,
  no.~4, pp. 128--135, Mar. 2021.

\bibitem{9019853}
P.~Wu, L.~Ding, Y.~Wang, L.~Li, H.~Zheng, and J.~Zhang, ``Performance analysis
  for uplink transmission in user-centric ultra-dense v2i networks,''
  \emph{IEEE Trans. Vehicular Technol.}, vol.~69, no.~9, pp. 9342--9355, Mar.
  2020.

\bibitem{Pervej_throughput}
M.~F. Pervej and S.-C. Lin, ``Dynamic power allocation and virtual cell
  formation for {Throughput-Optimal} vehicular edge networks in highway
  transportation,'' in \emph{Proc. IEEE ICC Workshops}, June 2020.

\bibitem{zhou2021user}
Y.~Zhou and Y.~Zhang, ``User-centric data communication service strategy for 5g
  vehicular networks,'' \emph{IET Com.}, Jul. 2021.

\bibitem{9462895}
Z.~Cheng, D.~Zhu, Y.~Zhao, and C.~Sun, ``Flexible virtual cell design for
  ultradense networks: A machine learning approach,'' \emph{IEEE Access},
  vol.~9, pp. 91\,575--91\,583, June 2021.

\bibitem{9134799}
H.~Xiao, X.~Zhang, A.~T. Chronopoulos, Z.~Zhang, H.~Liu, and S.~Ouyang,
  ``Resource management for multi-user-centric v2x communication in dynamic
  virtual-cell-based ultra-dense networks,'' \emph{IEEE Trans. Commun.},
  vol.~68, no.~10, pp. 6346--6358, Oct. 2020.

\bibitem{8334916}
T.~Sahin, M.~Klugel, C.~Zhou, and W.~Kellerer, ``Virtual cells for 5g v2x
  communications,'' \emph{IEEE Commun. Standards Mag.}, vol.~2, no.~1, pp.
  22--28, Mar. 2018.

\bibitem{8088603}
T.~Şahin, M.~Klügel, C.~Zhou, and W.~Kellerer, ``Multi-user-centric virtual
  cell operation for v2x communications in 5g networks,'' in \emph{Proc. CSCN},
  Sept. 2017.

\bibitem{Pervej_EE}
{M. F. {Pervej}} and {S.-C. {Lin}}, ``{Eco-Vehicular} edge networks for
  connected transportation: A distributed multi-agent reinforcement learning
  approach,'' in \emph{Proc. IEEE VTC2020-Fall}, Oct. 2020.

\bibitem{6994333}
D.~Kreutz, F.~M.~V. Ramos, P.~E. Veríssimo, C.~E. Rothenberg, S.~Azodolmolky,
  and S.~Uhlig, ``Software-defined networking: A comprehensive survey,''
  \emph{Proc. IEEE}, vol. 103, no.~1, pp. 14--76, Dec. 2015.

\bibitem{huang2021delay}
X.~Huang, K.~Xu, Q.~Chen, and J.~Zhang, ``Delay-aware caching in internet of
  vehicles networks,'' \emph{IEEE Internet Things J.}, 2021.

\bibitem{nan2021delay}
Z.~Nan, Y.~Jia, Z.~Ren, Z.~Chen, and L.~Liang, ``Delay-aware content delivery
  with deep reinforcement learning in internet of vehicles,'' \emph{IEEE Trans.
  Intel. Transport. Syst.}, 2021.

\bibitem{8993754}
S.~{Fang}, H.~{Chen}, Z.~{Khan}, and P.~{Fan}, ``On the content delivery
  efficiency of noma assisted vehicular communication networks with delay
  constraints,'' \emph{IEEE Wireless Commun. Letters}, pp. 1--1, Feb 2020.

\bibitem{8998330}
Y.~{Dai}, D.~{Xu}, K.~{Zhang}, S.~{Maharjan}, and Y.~{Zhang}, ``Deep
  reinforcement learning and permissioned blockchain for content caching in
  vehicular edge computing and networks,'' \emph{IEEE Trans. Vehicular
  Technol.}, vol.~69, no.~4, pp. 4312--4324, Feb. 2020.

\bibitem{8998397}
Y.~{Lu}, X.~{Huang}, K.~{Zhang}, S.~{Maharjan}, and Y.~{Zhang}, ``Blockchain
  empowered asynchronous federated learning for secure data sharing in internet
  of vehicles,'' \emph{IEEE Trans. Vehicular Technol.}, vol.~69, no.~4, pp.
  4298--4311, Feb. 2020.

\bibitem{zhang2020smart}
Z.~Zhang, C.-H. Lung, M.~St-Hilaire, and I.~Lambadaris, ``Smart proactive
  caching: Empower the video delivery for autonomous vehicles in icn-based
  networks,'' \emph{IEEE Trans. Vehicular Technol.}, May 2020.

\bibitem{9129007}
S.~{Fang}, Z.~{Khan}, and P.~{Fan}, ``A cooperative rsu caching policy for
  vehicular content delivery networks in two-way road with a t-junction,'' in
  \emph{Proc. VTC2020-Spring}, June 2020.

\bibitem{9417383}
W.~Liu, H.~Zhang, H.~Ding, D.~Li, and D.~Yuan, ``Mobility-aware coded edge
  caching in vehicular networks with dynamic content popularity,'' in
  \emph{Proc. IEEE WCNC}, 2021.

\bibitem{8531745}
Y.~{Jiang}, M.~{Ma}, M.~{Bennis}, F.~{Zheng}, and X.~{You}, ``User preference
  learning-based edge caching for fog radio access network,'' \emph{IEEE Trans.
  Commun.}, vol.~67, no.~2, pp. 1268--1283, Nov. 2018.

\bibitem{malik2020personalized}
A.~Malik, J.~Kim, K.~S. Kim, and W.-Y. Shin, ``A personalized preference
  learning framework for caching in mobile networks,'' \emph{IEEE Trans. Mobile
  Computing}, Feb. 2020.

\bibitem{9275345}
Y.~Lin, Z.~Zhang, Y.~Huang, J.~Li, F.~Shu, and L.~Hanzo, ``Heterogeneous
  user-centric cluster migration improves the connectivity-handover trade-off
  in vehicular networks,'' \emph{IEEE Trans. Vehicular Technol.}, vol.~69,
  no.~12, pp. 16\,027--16\,043, Dec. 2020.

\bibitem{lin2021sd}
S.-C. Lin, K.-C. Chen, and A.~Karimoddini, ``Sdvec: Software-defined vehicular
  edge computing with ultra-low latency,'' \emph{IEEE Commun. Magaz.}, vol.~59,
  no.~12, pp. 66--72, 2021.

\bibitem{3GPP_TR_38_886}
``\textit{3rd Generation Partnership Project; Technical Specification Group
  Radio Access Network; V2X Services based on NR; User Equipment (UE) radio
  transmission and reception},'' 3GPP TR 38.886 V0.5.0, Release 16, Feb. 2020.

\bibitem{graham1989concrete}
R.~L. Graham, D.~E. Knuth, O.~Patashnik, and S.~Liu, ``Concrete mathematics: a
  foundation for computer science,'' \emph{Computers in Physics}, vol.~3,
  no.~5, pp. 106--107, 1989.

\bibitem{3GPP_TR_38_901}
``3rd generation partnership project; technical specification group radio
  access network; study on channel model for frequencies from 0.5 to 100 ghz,''
  3GPP TR 38.901 V16.1.0, Release 16, Dec. 2019.

\bibitem{8374824}
S.~O. {Somuyiwa}, A.~{György}, and D.~{Gündüz}, ``A reinforcement-learning
  approach to proactive caching in wireless networks,'' \emph{IEEE J. Sel.
  Areas Commun.}, vol.~36, no.~6, pp. 1331--1344, 2018.

\bibitem{8357917}
B.~N. {Bharath}, K.~G. {Nagananda}, D.~{Gündüz}, and H.~V. {Poor}, ``Caching
  with time-varying popularity profiles: A learning-theoretic perspective,''
  \emph{IEEE Trans. Commun.}, vol.~66, no.~9, pp. 3837--3847, 2018.

\bibitem{wang1993number}
Y.~H. Wang, ``On the number of successes in independent trials,''
  \emph{Statistica Sinica}, pp. 295--312, 1993.

\bibitem{sutton2018reinforcement}
R.~S. Sutton and A.~G. Barto, \emph{Reinforcement learning: An
  introduction}.\hskip 1em plus 0.5em minus 0.4em\relax MIT press, 2018.

\bibitem{watkins1992q}
C.~J. Watkins and P.~Dayan, ``Q-learning,'' \emph{Machine learning}, vol.~8,
  no. 3-4, pp. 279--292, 1992.

\bibitem{mnih2015human}
V.~Mnih \emph{et~al.}, ``Human-level control through deep reinforcement
  learning,'' \emph{nature}, vol. 518, no. 7540, pp. 529--533, 2015.

\bibitem{8624268}
T.~Guo and A.~Suárez, ``Enabling 5g ran slicing with edf slice scheduling,''
  \emph{IEEE Trans. Vehicular Technol.}, vol.~68, no.~3, pp. 2865--2877, Mar.
  2019.

\bibitem{9348504}
P.~Guan and X.~Deng, ``Maximize potential reserved task scheduling for urllc
  transmission and edge computing,'' in \emph{Proc. VTC2020-Fall}, 2020.

\bibitem{buttazzo2011hard}
G.~C. Buttazzo, \emph{Hard real-time computing systems: predictable scheduling
  algorithms and applications}.\hskip 1em plus 0.5em minus 0.4em\relax Springer
  Science \& Business Media, 2011, vol.~24.

\bibitem{kuhn1955hungarian}
H.~W. Kuhn, ``The hungarian method for the assignment problem,'' \emph{Naval
  research logistics quarterly}, vol.~2, no. 1-2, pp. 83--97, 1955.

\bibitem{west2001introduction}
D.~B. West \emph{et~al.}, \emph{Introduction to graph theory}.\hskip 1em plus
  0.5em minus 0.4em\relax Prentice hall Upper Saddle River, 2001, vol.~2.

\bibitem{8569938}
P.~A. {Lopez}, M.~{Behrisch}, L.~{Bieker-Walz}, J.~{Erdmann}, Y.~{Flötteröd},
  R.~{Hilbrich}, L.~{Lücken}, J.~{Rummel}, P.~{Wagner}, and E.~{Wiessner},
  ``Microscopic traffic simulation using sumo,'' in \emph{Proc. ITSC}, Nov.
  2018.

\bibitem{roy2011handbook}
R.~R. Roy, \emph{Handbook of mobile ad hoc networks for mobility models}.\hskip
  1em plus 0.5em minus 0.4em\relax Springer, 2011, vol. 170.

\bibitem{wegener2008traci}
A.~Wegener, M.~Pi{\'o}rkowski, M.~Raya, H.~Hellbr{\"u}ck, S.~Fischer, and J.-P.
  Hubaux, ``Traci: an interface for coupling road traffic and network
  simulators,'' in \emph{Proc. CNSS}, 2008.

\bibitem{8809280}
Y.~Zhang, R.~Wang, M.~S. Hossain, M.~F. Alhamid, and M.~Guizani,
  ``Heterogeneous information network-based content caching in the internet of
  vehicles,'' \emph{IEEE Trans. Vehicular Technol.}, vol.~68, no.~10, pp.
  10\,216--10\,226, Aug. 2019.

\bibitem{cacheRepSurvey}
S.~Podlipnig and L.~B\"{o}sz\"{o}rmenyi, ``A survey of web cache replacement
  strategies,'' \emph{ACM Comput. Surv.}, vol.~35, no.~4, p. 374–398, dec
  2003.

\bibitem{kellerer2004multidimensional}
H.~Kellerer, U.~Pferschy, and D.~Pisinger, ``Multidimensional knapsack
  problems,'' in \emph{Knapsack problems}.\hskip 1em plus 0.5em minus
  0.4em\relax Springer, 2004, pp. 235--283.

\end{thebibliography}
\bibliographystyle{IEEEtran}

\end{document}